\documentclass[a4paper, UKenglish, cleveref, numberwithinsect]{false-lipics-v2021}

\usepackage[normalem]{ulem}
\usepackage{csquotes}
\usepackage[bibliography=common]{apxproof}

\usepackage{multicol}
\usepackage{tikz}
\usepackage{mathtools}
\usepackage{adjustbox}
\usepackage{nth}
\usetikzlibrary{arrows.meta,calc,backgrounds,shapes.misc,decorations.pathreplacing}

\definecolor{Wonggold}{HTML}{E69F00}
\definecolor{Wonglightblue}{HTML}{56B4E9}
\definecolor{Wongteal}{HTML}{009e73}
\definecolor{Wongyellow}{HTML}{f0e442}
\definecolor{Wongblue}{HTML}{0072b2}
\definecolor{Wongorange}{HTML}{d55e00}
\definecolor{Wongpink}{HTML}{cc79a7}

\colorlet{darkWongorange}{Wongorange!90!black}
\colorlet{darkWongblue}{Wongblue!90!black}
\colorlet{darkWongpink}{Wongpink!90!black}
\colorlet{darkWongyellow}{Wongyellow!90!black}
\colorlet{darkWongteal}{Wongteal!90!black}

\definecolor{ibmblue}{HTML}{648FFF}
\definecolor{ibmpurple}{HTML}{785EF0}
\definecolor{ibmpink}{HTML}{DC267F}
\definecolor{ibmorange}{HTML}{FE6100}
\definecolor{ibmgold}{HTML}{FFB000}

\definecolor{Tolpurple}{HTML}{332288}
\definecolor{Tollightblue}{HTML}{88CCEE}
\definecolor{Tolteal}{HTML}{44AA99}
\definecolor{Tolgreen}{HTML}{117733}
\definecolor{Tololive}{HTML}{999933}
\definecolor{Tolyellowish}{HTML}{DDCC77}
\definecolor{Tolpink}{HTML}{CC6677}
\definecolor{Tolmagenta}{HTML}{AA4499}

\colorlet{outercolor}{Wongorange}
\colorlet{innercolor}{Wongblue}
\colorlet{finitarycolor}{Wongteal!90!green}

\AddToHook{env/property/begin}{\crefalias{theorem}{property}}
\AddToHook{env/proposition/begin}{\crefalias{theorem}{proposition}}
\AddToHook{env/lemma/begin}{\crefalias{theorem}{lemma}}
\AddToHook{env/example/begin}{\crefalias{theorem}{example}}
\AddToHook{env/definition/begin}{\crefalias{theorem}{definition}}
\AddToHook{env/corollary/begin}{\crefalias{theorem}{corollary}}
\AddToHook{env/claim/begin}{\crefalias{theorem}{claim}}
\AddToHook{env/remark/begin}{\crefalias{theorem}{remark}}
\AddToHook{env/conjecture/begin}{\crefalias{theorem}{conjecture}}

\newif\ifcomments\commentstrue \newcommand{\mcomment}[3][blue]{\ifcomments\ifmmode\margincomment[#1]{#2}{#3}\else\footcomment[#1]{#2}{#3}\fi\fi}
\newcommand{\footcomment}[3][blue]{\ifcomments{\small\textcolor{#1}{\textbf{(#2)}}}\footnote{\textcolor{#1}{\textbf{#2:} #3}}\fi}
\newcommand{\margincomment}[3][blue]{\ifcomments{\color{#1}\textbf{(#1)}}\footnotemark\marginnote{\tiny\textsuperscript{\thefootnote}\textbf{#1:} #2}\fi}
\newcommand{\intextcomment}[3][blue]{\ifcomments\textcolor{#1}{\textbf{#2:} #3}\fi}
\newcommand{\inremarkcomment}[3][blue]{\ifcomments{\color{#1}\begin{remark}[#2]#3\end{remark}}\fi}
\newcommand{\framedcomment}[3][blue]{\ifcomments\par\medskip\noindent{\fboxrule=1pt\color{#1}\fbox{\begin{minipage}{\dimexpr\linewidth-2\fboxsep-2\fboxrule\relax}\parindent=0pt\stepcounter{footnote}\textbf{(\thefootnote)\,#2}: #3\end{minipage}}}\par\medskip\fi}

\makeatletter
\newcommand{\register@comment}[3][blue]{\colorlet{#2color}{#1}    \global\expandafter\def\csname#2\endcsname{\@ifstar{\intextcomment[#1]{#3}}{\mcomment[#1]{#3}}}
    \global\expandafter\def\csname#3\endcsname{\@ifstar{\inremarkcomment[#1]{#3}}{\framedcomment[#1]{#3}}}
}
\makeatother

\newcommand{\defemph}[1]{{\color{darkblue}\textbf{#1}}}

\renewcommand{\b}[1]{\texttt{#1}}

\colorlet{darkorange}{orange!66!black}
\colorlet{darkblue}{blue!75!black}

\newcommand{\ppreleqnostar}[1]{\mathrel{\adjustbox{raise=-.15ex}{$\overline{\mathsf{PP}}_{#1}$}}}
\newcommand{\ppreleqstar}[1]{\mathrel{\adjustbox{raise=-.15ex}{$\overline{\mathsf{PP}}^{\,*}_{#1}$}}}
\makeatletter
\newcommand{\ppreleq}{\@ifstar{\ppreleqstar}{\ppreleqnostar}}
\makeatother

\newcommand{\pprelnostar}[1]{\mathrel{\mathsf{PP}_{#1}}}
\newcommand{\pprelstar}[1]{\mathrel{\mathsf{PP}^{\,*}_{#1}}}
\makeatletter
\newcommand{\pprel}{\@ifstar{\pprelstar}{\pprelnostar}}
\makeatother

\DeclareMathOperator{\src}{\mathsf{Src}}
\DeclareMathOperator{\tgt}{\mathsf{Tgt}}

\DeclareMathOperator{\peri}{\mathsf{Peri}}
\DeclareMathOperator{\crit}{\mathsf{Crit}}
\DeclareMathOperator{\evperi}{\mathsf{EvPeri}}

\DeclareMathOperator{\Walks}{\mathsf{Walks}}
\newcommand{\walks}[2][]{\Walks_{#1}(#2)}

\newcommand{\N}{\mathbb{N}}

\makeatletter
\newcommand{\R}[1]{\@ifstar{\mathbin{{R_{#1}}^{\hspace*{-.5ex}*}}}{\mathbin{R_{#1}}}}
\makeatother

\newcommand{\walkind}[1]{\ensuremath{\mathsf{#1}}}
\renewcommand{\L}{\walkind{L}}
\newcommand{\I}{\walkind{I}}
\renewcommand{\R}{\walkind{R}}
\newcommand{\F}{\walkind{F}}

\newtheoremrep{theorem}{Theorem}[section]
\newtheoremrep{proposition}[theorem]{Proposition}

\newtheoremrep{property}[theorem]{Property}
\newtheoremrep{lemma}[theorem]{Lemma}
\crefname{claim}{Claim}{Claims}
\crefname{property}{Property}{Properties}
\crefname{properties}{Properties}{Properties}
\crefname{openproblem}{Open problem}{Open problems}

\newtheorem{openproblem}[theorem]{Open problem}

\theoremstyle{claimstyle}
\newtheorem{falseclaiminner}{Claim}
\newenvironment{falseclaim}[1]{\renewcommand{\thefalseclaiminner}{#1}\begin{falseclaiminner}}{\end{falseclaiminner}}

\title{Graphs with Bounded Push-Pop Walk Orderings}

\author{Antoine Amarilli}
{LTCI, Télécom Paris, Institut Polytechnique de Paris, France
\and Univ.~Lille, Inria, CNRS, Centrale Lille, UMR 9189 CRIStAL, F-59000 Lille, France
\and \url{https://a3nm.net}}
       {antoine.amarilli@telecom-paris.fr}
       {https://orcid.org/0000-0002-7977-4441}
       {}
\author{Claire David}
       {Univ.~Gustave Eiffel, CNRS, LIGM, F-77454 Marne-la-Vallée, France}
       {claire.david@univ-eiffel.fr}
       {https://orcid.org/0000-0002-3041-7013}
       {}
\author{Nadime Francis}
       {Univ.~Gustave Eiffel, CNRS, LIGM, F-77454 Marne-la-Vallée, France}
       {nadime.francis@univ-eiffel.fr}
       {https://orcid.org/0009-0009-4531-7435}{}
\author{Victor Marsault}
       {Univ.~Gustave Eiffel, CNRS, LIGM, F-77454 Marne-la-Vallée, France
        \and \url{https://victor.marsault.xyz}}
       {victor.marsault@univ-eiffel.fr}
       {https://orcid.org/0000-0002-2325-6004}
       {}
\author{Mikaël Monet}
       {Univ.~Lille, Inria, CNRS, Centrale Lille, UMR 9189 CRIStAL, F-59000 Lille, France}
       {mikael.monet@inria.fr}
       {https://orcid.org/0000-0002-6158-4607}
       {}
\author{Yann Strozecki}
       {Université Paris-Saclay, UVSQ, Données et Algorithmes pour une ville intelligente et durable, 78035, Versailles, France.}
       {yann.strozecki@uvsq.fr}
       {https://orcid.org/0000-0002-0891-3766}
       {}

\authorrunning{A.~Amarilli, C.~David, N.~Francis, V.~Marsault, M.~Monet and Y.~Strozecki}

\Copyright{Antoine Amarilli, Claire David, Nadime Francis, Victor Marsault, Mikaël Monet, and Yann Strozecki}

\makeatletter
\let\c@author\relax
\makeatother
\usepackage[,sorting=anyt,minalphanames=2,maxalphanames=3,maxbibnames=100,giveninits=true,style=alphabetic]{biblatex}

\bibliography{biblio}

\keywords{Push-pop Distance, Walk Ordering, Walk Enumeration, Combinatorial Gray Code, Graphs} 

\ccsdesc[500]{Mathematics of computing~Enumeration}
\ccsdesc[300]{Mathematics of computing~Graph theory}

\acknowledgements{Antoine Amarilli is partially supported by the project
  ANR-25-CE23-1215 (``Expand''); Victor Marsault and Nadime Francis by ANR-21-CE48-0015 (``VeriGraph'').}

\nolinenumbers

\hideLIPIcs

\begin{document}

\maketitle

\begin{abstract}
    Consider a directed graph~$G$ and the (generally infinite) set~$\walks{G}$ of the walks
    in that graph. 
We are interested in ordering $\walks{G}$ as $w_1, w_2,
    \ldots$ in a way that ensures that each walk is obtained from the previous one
    by a \emph{bounded} number of elementary operations, namely, 
    \emph{push operations} (which add a new edge at  either end of
    the walk) and \emph{pop operations} (which remove the first or last edge of the
    walk). 
Our goal is to optimize the aforementioned bound, when it exists,
    in particular when imposing the additional requirement that walks must be ordered in increasing order of length.
    
    To this end, we introduce a condition on graphs inspired by earlier work, called
    \emph{inner-connectedness},
    which is testable in linear time; we show that
    inner-connected graphs are precisely those which admit bounded walk
    orderings.
Focusing on
    such graphs, we then introduce and study the \emph{pp-rank} and the
    \emph{pp-increasing-rank}. Namely, the
    \emph{pp-rank} of a graph is the best
    possible bound on the number of push-pop operations in
    a bounded walk ordering for that graph; and the \emph{pp-increasing-rank}
    is the same where we further impose that shorter walks appear before longer walks in the ordering.
We show that
    the pp-rank of inner-connected graphs is in fact always bounded by~$3$. The pp-increasing-rank of inner-connected graphs is also always finite,
    but it can be up to linear in the
    input graph.
More precisely, we show that the pp-increasing-rank is related to a new parameter we introduce, the \emph{peripherality} of graphs.
    In particular, for cyclic graphs, peripherality is related to the pp-increasing-rank up to 
multiplicative
    and additive constants. 
We last study how to compute these parameters, and
    show that while peripherality can be computed in polynomial time,
    the pp-increasing-rank itself is generally NP-hard to compute.
\end{abstract}

\vskip\topmattervskip\baselineskip\noindent\noindent\textcolor{lipicsGray}{\fontsize{9}{12}\sffamily\bfseries Usage of AI\enskip} {\fontsize{9}{12}\selectfont We used LLMs for two purposes: bibliographic search and finding examples to build our intuition or test our conjectures.  In particular, we found \cite{LaiWei1993} this way, which turned out to be crucial in the proof of \cref{p:4ord-nphard}.  
  In terms of examples, the only one in the final paper which was found by an LLM is given in \cref{f:1-orderable-dag}, in Appendix~\ref{a:open-problems}.} 

\newpage

\section{Introduction}

Efficient enumeration of combinatorial objects is a classical problem in
theoretical computer science, with applications in exhaustive search,
combinatorial generation and database query evaluation.
The efficiency of enumeration algorithms is usually measured by its
\emph{delay}, i.e., the maximal time between consecutive
outputs.
While constant-delay is ideal~\cite{segoufin2014glimpse},
it requires writing each solution in constant time, which is unrealistic
when solution size depends on input size, or when there are infinitely many solutions.
To bound the delay in such cases, we adopt the approach of \emph{combinatorial Gray codes}~\cite{Mutze23-gray}, where 
objects (e.g., binary strings, permutations or spanning trees) must be listed 
exactly once, with consecutive objects differing by a small change. 
In this setting, the time
to write the output is not accounted for, following the ``don't count the
output'' principle of combinatorial generation~\cite[p.\,8]{ruskey2003combinatorial}.

In this paper, we apply this paradigm to the enumeration of all walks in a directed graph. 
Our motivation comes from graph databases, where modern query languages (Cypher~\cite{FrancisEtAl2018} or ISO standards GQL~\cite{GQL-ISO} and SQL/PGQ~\cite{SQLPGQ-ISO})
are based on extracting all walks satisfying navigational constraints (\texttt{MATCH}).
Walks are usually computed and output in full, 
which is sometimes the
bottleneck for some tasks, e.g., enumerating all \emph{shortest} walks satisfying the constraints (\texttt{MATCH SHORTEST}, \cite{david2024distinct}).
Output handling is identified as one of the current research challenges in the
field \cite[Sec.\,7.1]{LibkinEtAl25-gems}: 
potential solutions include factorized representations \cite{MartensEtAl23,MartensEtAl24} and outputting only the difference between consecutive outputs.
Here, we 
study the latter approach in the simplest setting: unlabeled graphs with no
constraint on walks, where we 
enumerate the (possibly infinite) set of all walks.
We require each produced walk to be
at a bounded \emph{push-pop distance} from the previous one, i.e., 
to be obtained by applying at most~$d$ \emph{push operations} (add an edge to
either end of the walk) or \emph{pop operations} (remove the first or last edge of the walk), for some~$d\in \mathbb{N}$.

Our first goal is to classify which graphs are \emph{orderable}, i.e., admit
such a bounded walk ordering, and then to quantify the best possible distance
that can be achieved: we call this the \emph{pp-rank} of the graph.
Note that some graphs are not orderable, as shown in
\cref{f:inner-disconnected}; whereas others can be shown to be orderable, in
\cref{f:inner-connected}
(ignoring all edge colors and numbers for now).
Next, motivated by 
shortest walk enumeration in graph databases,
our second goal is to classify which graphs are \emph{orderable by increasing
length}, i.e., admit a bounded ordering with walks occurring in increasing order of length.
The minimal achievable distance in that case is called \emph{pp-increasing-rank} of the graph. 
Last, since it is sometimes harder to order ``short'' walks than ``long'' walks,
our third goal is to study the \emph{eventual pp-increasing-rank}, i.e.,
the smallest possible distance $d$ for which we have a walk ordering by
increasing length that 
satisfies bound~$d$, except for finitely many walks.
Note that
graphs are orderable by increasing length 
iff they are \emph{eventually} orderable by increasing length, but
their eventual pp-increasing-rank is generally smaller than their pp-increasing-rank.

\subparagraph*{Contributions.}
Our first contribution is to provide a linear-time testable characterization of graph orderability and to establish a distance bound of~$3$ for orderable graphs.
To this end, we use \emph{inner-connectedness}, inspired by a notion from~\cite{amarilli2023enumerating} (see the \emph{related work} paragraph for details). 
Their results imply that inner-\textbf{dis}connected graphs are \textbf{not} orderable.
Conversely, we show that inner-connected graphs have a pp-rank of at most~$3$, i.e., each admits a walk ordering where consecutive walks differ by at most $3$ operations.

Our second and main contribution is to study 
orderability by increasing length.
Inner-disconnected graphs are not orderable by increasing length because they
are not orderable at all.
We show that, conversely, all inner-connected graphs
are orderable by increasing length,
with a pp-increasing-rank bound which is
generally linear in the graph size.

To this end, we introduce
\emph{peripherality}
as a graph parameter that allows us to approximate the
pp-increasing-rank.
Informally, the \emph{peripherality} of a walk is the
smallest number of edges to remove from the walk in order to be able to reach a
nontrivial scc from one of its endpoints. 
The \emph{peripherality} of a graph~$G$, denoted by~$\peri(G)$, is
the maximal peripherality among its walks, and its \emph{eventual peripherality}, denoted by~$\evperi(G)$, is the
maximal peripherality achieved by infinitely many walks. 
We show that the pp-increasing and eventual pp-increasing-rank of a cyclic graph~$G$ are in~$\Theta(\peri(G))$ and~$\Theta(\evperi(G))$, respectively.
By contrast, the pp-increasing-rank of an acyclic graph is in~$O(\peri(G))$ but
there are acyclic graphs with arbitrary peripherality but constant pp-increasing-rank.

Our results on inner-connected graphs are summarized in the following theorem.
We exclude trivial graphs with zero
or one edge, as they do not follow the general bounds.

\begin{theorem}\label{t:pp-rank}
    Let~$G=(V,E)$ be an inner-connected graph with at least two edges.    \begin{enumerate}[a.]
        \item \label{t:pp-rank:classic}The pp-rank of~$G$ belongs to $[1,3]$.
        \item \label{t:pp-rank:incr}The pp-increasing-rank of~$G$ belongs
          to $[2,6k]$ if $G$ is acyclic, and to $[k,6k]$ otherwise, 
where~$k=\max(1,\peri(G))$.
        \item \label{t:pp-rank:eve-incr} The eventual pp-increasing-rank of $G$ is $0$ if $G$ is acyclic, 
        it is $1$ if the only cycle in $G$ is a self-loop on a vertex with no other incident edge, 
        and
it belongs to $[2k',6k']$ otherwise, where~$k'=\max(1,\evperi(G))$.
        \item \label{t:pp-rank:param}We have~$\evperi(G) \leq \peri(G) < |V|$ and~$\evperi(G) < \frac{|V|}{2}$.
    \end{enumerate}
\end{theorem}

\begin{figure}
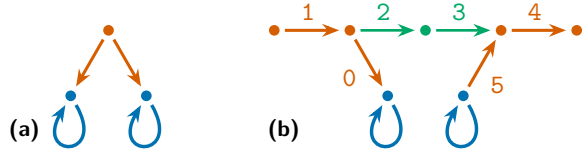

\centering
    \subfloat[]{\label{f:inner-disconnected1}\hspace*{2ex}\begin{tikzpicture}
    \newcommand{\dist}{10mm}
    \input{figure/graphstyle}
    
    \path node[vertex,inner] (v1) {} ++(60:\dist) 
          node[vertex,outer] (v2) {} ++(-60:\dist) 
          node[vertex,inner] (v3) {};

    \path[edge,outer] (v2) to (v1);
    \path[edge,outer] (v2) to (v3);
    \path[edge,inner,south loop] (v1) to (v1);
    \path[edge,inner,south loop] (v3) to (v3);
\end{tikzpicture}\vspace{-2.2em}}
    \hspace*{1cm}\subfloat[]{\label{f:inner-disconnected2}\begin{tikzpicture}
    \newcommand{\dist}{10mm}
    \input{figure/graphstyle}
    
    \path node[vertex,outer] (v1) {} ++(0:\dist) 
          node[vertex,outer] (v2) {} ++(0:\dist) 
          node[vertex,finitary] (v3) {} ++(0:\dist) 
          node[vertex,outer] (v4) {} ++(0:\dist) 
          node[vertex,outer] (v5) {};
    \path (v4) ++(-120:\dist) node[vertex,inner] (v0) {};
    \path (v2) ++(-60:\dist) node[vertex,inner] (v6) {};

    \path[edge,inner,south loop] (v0) to (v0);
    \path[edge,outer] (v0) to node[pos=.425,below right] {\b5} (v4);
    \path[edge,outer] (v1) to node[pos=.425,above] {\b1}  (v2);
    \path[edge,finitary] (v2) to node[pos=.425,above] {\b2}  (v3);
    \path[edge,finitary] (v3) to node[pos=.425,above] {\b3}  (v4);
    \path[edge,outer] (v4) to node[pos=.425,above] {\b4} (v5);
    \path[edge,outer] (v2) to node[pos=.425,below left] {\b0}  (v6);
    \path[edge,inner,south loop] (v6) to (v6);
\end{tikzpicture}\vspace{-2.2em}}
    \caption{Two graphs that are inner-disconnected, hence not orderable}
    \label{f:inner-disconnected}
\end{figure}
\begin{figure}
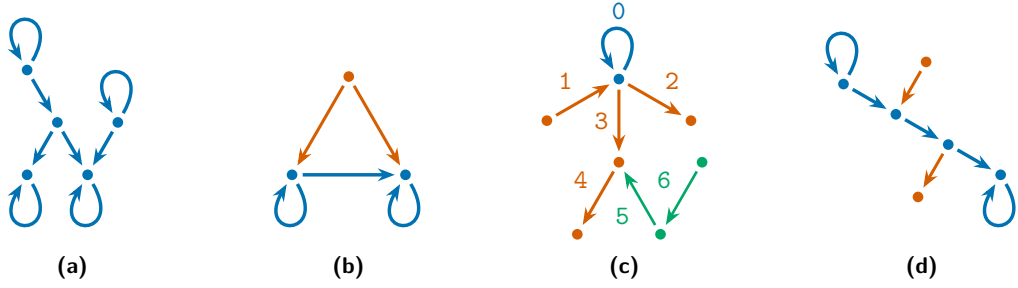

    \captionsetup[subfigure]{justification=centering, singlelinecheck=false}
    \newcommand{\connFigScale}{1}
    \subfloat[]{\scalebox{\connFigScale}{\begin{tikzpicture}
    \newcommand{\dist}{8mm}
    \input{figure/graphstyle}
    
\path node[vertex,inner] (v2) {} ++(60:\dist) 
          node[vertex,inner] (v3) {} ++(-60:\dist) 
          node[vertex,inner] (v4) {} ++(60:\dist) 
          node[vertex,inner] (v5) {};
    \path(v3) ++(120:\dist) node[vertex,inner] (v1){};

    \path[edge,inner] (v1) to (v3);
    \path[edge,inner] (v3) to (v2);
    \path[edge,inner] (v3) to (v4);
    \path[edge,inner] (v5) to (v4);

    \path[edge,inner,north loop] (v1) to (v1);
    \path[edge,inner,north loop] (v5) to (v5);
    \path[edge,inner,south loop] (v2) to (v2);
    \path[edge,inner,south loop] (v4) to (v4);
    
\end{tikzpicture}}\label{f:inner-connected:z}}\hfill
    \subfloat[]{\scalebox{\connFigScale}{\begin{tikzpicture}
    \newcommand{\dist}{15mm}
    \input{figure/graphstyle}
    
    \path node[vertex,inner] (v1) {} ++(60:\dist) 
          node[vertex,outer] (v2) {} ++(-60:\dist) 
          node[vertex,inner] (v3) {};

    \path[edge,outer] (v2) to (v1);
    \path[edge,outer] (v2) to (v3);
    \path[edge,inner] (v1) to (v3);
    \path[edge,inner,south loop] (v1) to (v1);
    \path[edge,inner,south loop] (v3) to (v3);
\end{tikzpicture}}\label{f:inner-connected:d}}\hfill
    \subfloat[]{\scalebox{\connFigScale}{\begin{tikzpicture}\newcommand{\dist}{11mm}\input{figure/graphstyle}
    
    \path node[vertex,inner] (v1) {} ++(-90:\dist) 
          node[vertex,outer] (v2) {} ++(-120:\dist) 
          node[vertex,outer] (v3) {};
    \path (v2)  ++(-60:\dist) node[vertex,finitary] (v4) {}
                ++(60:\dist) node[vertex,finitary] (v7) {}
    ;
    \path (v1) ++(-150:\dist) node[vertex,outer] (v5) {};
    \path (v1) ++(-30:\dist) node[vertex,outer] (v6) {};

    \path[edge,inner,north loop] (v1) to node[midway, above] {\b0} (v1);
    \path[edge,outer] (v5) to node [midway,above left] {\b1} (v1);
    \path[edge,outer] (v1) to node [midway,above right] {\b2} (v6);
    \path[edge,outer] (v1) to node[midway,left] {\b3} (v2);
    \path[edge,outer] (v2) to node[midway,above left] {\b4} (v3);
    \path[edge,finitary] (v4) to node[midway,below left] {\b5} (v2);
    \path[edge,finitary] (v7) to node[midway,above left] {\b6} (v4);
\end{tikzpicture}}\label{f:inner-connected:man}}\hfill
    \subfloat[]{\scalebox{\connFigScale}{\begin{tikzpicture}\newcommand{\dist}{8mm}\input{figure/graphstyle}
    
    \path node[vertex,inner] (v1) {} ++(-30:\dist) 
          node[vertex,inner] (v2) {} ++(-30:\dist) 
          node[vertex,inner] (v3) {} ++(-30:\dist) 
          node[vertex,inner] (v4) {};
    \path (v2) ++(60:\dist) node[vertex,outer] (v5) {};
    \path (v3) ++(-120:\dist) node[vertex,outer] (v6){};

    \path[edge,inner] (v1) to (v2);
    \path[edge,inner] (v2) to (v3);
    \path[edge,inner] (v3) to (v4);
    \path[edge,outer] (v5) to (v2);
    \path[edge,outer] (v3) to (v6);
    \path[edge,inner,north loop] (v1) to (v1);
    \path[edge,inner,south loop] (v4) to (v4);
\end{tikzpicture}}\label{f:inner-connected:x}}\caption{Four graphs that are inner-connected, hence orderable by increasing length}
    \label{f:inner-connected}
\end{figure}

The upper bounds on peripherality given in \cref{t:pp-rank:param} are virtually tight: one can easily find infinite families of graphs with
peripherality $\Omega(|V|)$ 
(see, e.g., \cref{f:high-peripherality}); and the same remark holds for eventual peripherality.
We 
also construct infinite graph families achieving three of the four lower bounds in \cref{t:pp-rank:classic,t:pp-rank:incr,t:pp-rank:eve-incr} of \cref{t:pp-rank} exactly, with the fourth off by one (\cref{p:tight}).
Finally, we show that peripherality and eventual peripherality can be
computed in polynomial time (\cref{p:peri-comp,p:crit-comp}),
but
computing
the pp-increasing-rank is NP-hard (\cref{p:4ord-nphard}).

\subparagraph{Related work.} 
The enumeration of paths and walks in graphs has been studied across multiple areas.
For instance, \cite{Eppstein98-compact-shortest} studies
algorithms to enumerate the $k$ shortest walks between two vertices in a
digraph.
In database theory, \cite{martens2018evaluation} studies efficient
enumeration of simple paths and trails that match regular expressions, and \cite{LiangEtAl2024} studies enumeration of
simple paths in graphs satisfying a restricted kind of regular expressions.
More relevant to us, \cite{adamson2024enumerating} studies constant-delay
enumeration of walks that have length~$m$ or length between
$\ell$ and~$m$. 
In their case, constant delay relies on edits to a compressed representation of the walks rather than push-pop edits on one explicit walk.

The main originality of our work
is to address two goals: (1.) producing walks by performing
a constant number of modifications on the walks (specifically via push-pop
edits); and (2.) ordering the infinite set of
all walks (rather than simple paths, or walks in a specific length range).
For goal (1.), push-pop operations are considered 
more generally
in~\cite{chung1992universal}, and enumeration with a constant push-pop distance specifically appears in 
\cite{CDCASGrayCode-icalp}.
More precisely, \cite{CDCASGrayCode-icalp} studies the enumeration of all binary words
of length $\ell$ via push-pop operations, and can be 
extended to enumerate
all binary words by increasing length (corresponding to goal (2.)).
Since binary words correspond to walks in a single-vertex graph with
two self-loops labeled 0 and 1, our work
generalizes~\cite{CDCASGrayCode-icalp} to 
arbitrary directed graphs;
by contrast, we only focus on the \emph{existence} of suitable orders and not on
\emph{computing} them in computational models with stringent time and space bounds like
in~\cite{CDCASGrayCode-icalp}.

The work most closely related to ours
is~\cite{amarilli2023enumerating}.
It studies how to enumerate the set~$L$ of the words accepted by a deterministic
finite automaton~$A$, while obeying a bound on the push-pop distance between consecutive words.
Together with~\cite{CDCASGrayCode-icalp}, this is (to our knowledge) the only
existing work that addresses goals (1.) and (2.).
Specifically, a linear-time-testable
criterion on~$A$ called \emph{interchangeability} is presented
in~\cite{amarilli2023enumerating} to decide whether a bounded ordering of~$L$ exists.
Since any directed graph can be converted
into a deterministic automaton (where walks in the graph correspond to the
words accepted by the automaton), the results of~\cite{amarilli2023enumerating}
also apply to graphs. 
Thus, our notion of \emph{inner-connectedness} specializes their
\emph{interchangeability} criterion, and their results directly imply that
inner-disconnected graphs do not admit bounded walk orderings
(\cref{p:inner-disconnected=>nonorderable}).
The novelty in our work relative to~\cite{amarilli2023enumerating} is twofold.
First, by focusing on walks in graphs,
we obtain a more favorable upper bound on 
pp-rank:
\cref{t:pp-rank:classic} of \cref{t:pp-rank} shows a bound of~$3$,
while 
their construction would give a bound of $48 (|V|+1)^2$, in fact
$48|Q|^2$ where~$Q$ is the automaton state set.
Second, our main contributions are about ordering walks \emph{by increasing length}, which is
not studied by~\cite{amarilli2023enumerating} and simpler in our setting (e.g., the language $a(aa)^* + b(bb)^* + ((a+b)(a+b))^*$ is orderable but not orderable by increasing length).

\subparagraph*{Outline.}
The rest of the paper is organized as follows.
Section~\ref{s:preliminaries}
introduces preliminary notions (graphs, walks, orderings, pp-ranks).
We then introduce terminology in \cref{s:typology} on object kinds and walk types, which we use to define \emph{inner-connectedness} and
restate the result
from~\cite{amarilli2023enumerating} 
that inner-disconnected graphs do not admit bounded walk orderings.
We then focus in \cref{s:peripherality,s:criticality} on enumerating walks by
increasing order of length and on the pp-increasing rank.
Specifically, in 
\cref{s:peripherality}, after defining \emph{peripherality} and \emph{eventual peripherality}, we establish the upper bounds of~\cref{t:pp-rank},~\cref{t:pp-rank:classic,t:pp-rank:incr,t:pp-rank:eve-incr}, as well as \cref{t:pp-rank:param}.
We then prove in \cref{s:criticality} the lower bounds
of~\cref{t:pp-rank},~\cref{t:pp-rank:incr,t:pp-rank:eve-incr}, by introducing
the \emph{criticality} and \emph{eventual criticality} graph parameters, and exhibit graph families
that (nearly) meet the lower bounds in~\cref{t:pp-rank:classic,t:pp-rank:incr,t:pp-rank:eve-incr}.
Section~\ref{s:param-comp} addresses computational aspects, showing that
peripherality and eventual peripherality are tractable to compute while pp-increasing-rank is not.
Finally, we  discuss extensions and future directions in
\cref{s:conclusion}.

Due to space constraints, the details of most proofs are deferred to Appendix~\ref{a:typology}--\ref{a:param-comp}.
We also provide some open problems in Appendix~\ref{a:open-problems}.

\begin{toappendix}
    \section{Open Problems}
\label{a:open-problems}
In this section, we implicitly assume that graph families are infinite, and that they contain graphs with arbitrarily many edges. (Since we do not assume graphs to be weakly connected, a graph can contain arbitrarily many vertices that are isolated and hence irrelevant.)

\subsection{Lower Bounds}

\begin{openproblem}
    Find a  family of 1-orderable acyclic graphs, or show that there is none.
\end{openproblem}

We found several 1-orderable acyclic graphs, but none of them has more than twelve edges. We provide in \cref{f:1-orderable-dag} a smaller example together with its walk graph.  Given the constraints imposed on a 1-bounded ordering for a single path, we deem it likely that there are only finitely many acyclic graphs that are 1-orderable.

\begin{figure}
    \centering

    \subfloat[]{\begin{tikzpicture}
    \newcommand{\dist}{15mm}
    \input{figure/graphstyle}

    \path node[vertex,finitary] (v0) {} 
        ++(30:\dist)  node[vertex,finitary] (v1) {}
        ++(-90:\dist) node[vertex,finitary] (v2) {}
        ++(30:\dist)  node[vertex,finitary] (v3) {}
        ++(-90:\dist)  node[vertex,finitary] (v4) {};
    
    \path[edge,finitary] (v0) to node[above left] {$a$} (v1);
    \path[edge,finitary] (v0) to node[below left] {$b$} (v2);
    \path[edge,finitary] (v1) to node[above left] {$c$} (v2);
    \path[edge,finitary] (v1) to node[above right] {$d$} (v3);
    \path[edge,finitary] (v2) to node[below right] {$e$} (v3);
    \path[edge,finitary] (v2) to node[below left] {$f$} (v4);
\end{tikzpicture}\vspace{-1.7em}}
    \hfill
    \subfloat[]{\begin{tikzpicture}
    \newcommand{\dist}{15mm}
    \input{figure/graphstyle}

    \path node (eps) {$\varepsilon$}
        ++(-120:\dist) node[walk] (c) {$c$}
        ++(0:\dist) node[walk] (e) {$e$}
        ++(0:\dist) node[walk] (b) {$b$}
        ++(0:\dist) node[walk] (f) {$f$};
    \path (c) 
        ++(180:\dist) node[walk] (a) {$a$}
        ++(180:\dist) node[walk] (d) {$d$};
    \path (d) ++(-60:\dist) node[walk] (ad) {$ad$};
    \path (a) ++(-60:\dist) node[walk] (ac) {$ac$};
    \path (c) ++(-60:\dist) node[walk] (ce) {$ce$};
    \path (e) ++(-60:\dist) node[walk] (be) {$be$};
    \path (b) ++(-60:\dist) node[walk] (bf) {$bf$};
    \path (ac) ++(-60:\dist) node[walk] (ace) {$ace$};
    \path (bf) ++(-60:\dist) node[walk] (acf) {$acf$};
    \path (f) ++(-60:\dist) node[walk] (cf) {$cf$};

    \path[walk edge] (eps) to (a);
    \path[walk edge] (eps) to (b);
    \path[walk edge, Wongorange] (eps) to (c);
    \path[walk edge] (eps) to (d);
    \path[walk edge, Wongorange] (eps) to (e);
    \path[walk edge] (eps) to (f);
    \path[walk edge, Wongorange] (a) to (ad);
    \path[walk edge, Wongorange] (d) to (ad);
    \path[walk edge, Wongorange] (a) to (ac);
    \path[walk edge] (c) to (ac);
    \path[walk edge] (c) to[out=-30,in=150] (cf);
    \path[walk edge, Wongorange] (c) to (ce);
    \path[over walk edge] (e) to (ce);
    \path[over walk edge, double=Wongorange] (b) to (be);
    \path[over walk edge, double=Wongorange] (e) to (be);
    \path[walk edge, Wongorange] (ac) to (ace);
    \path[walk edge] (ac) to[out=-30,in=150] (acf);
    \path[walk edge, Wongorange] (cf) to (acf);
    \path[walk edge, Wongorange] (f) to (cf);
    \path[over walk edge, double=Wongorange] (f) to (bf);
    \path[over walk edge, double=Wongorange] (b) to (bf);
    \path[over walk edge, double= Wongorange] (ce) to (ace);

\end{tikzpicture}\vspace{-2.0em}}
    
    \caption{A 1-orderable acyclic graph with 6 edges (a) and its walk graph (b) admitting a Hamiltonian path (orange)}
    \label{f:1-orderable-dag}
\end{figure}

\begin{openproblem}\label{op:peripherality-orderable}
    Find a family of cyclic graphs~$(G_k)_{k\in\N}$ where~$G_k$ has a peripherality of~$k$ and is $k$-orderable by increasing length, or show that there is none.
\end{openproblem}

We believe there is no such family and that the pp-increasing rank of a cyclic graph~$G$ is at least~$\peri(G)+1$.

\subsection{Upper Bounds}

Unlike lower bounds, we have not shown our upper bounds to be tight at all. Hence there is much work to do in that direction.

\begin{openproblem}\label{op:3-orderable}
    Find a graph (or a family of graphs) with a pp-rank of 3,
    or show that there is none.
\end{openproblem}

We tend to think that there are indeed graphs that are not 2-orderable but we suspect that all trees, and possibly all acyclic graphs, are 2-orderable.
Note that most graphs are not 1-orderable and it can be verified for instance
that the graph consisting of a single directed path is not 1-orderable as soon as it has five edges.

\begin{openproblem} Improve the upper bounds of \cref{t:pp-rank},  \cref{t:pp-rank:incr,t:pp-rank:eve-incr} or find a graph family showing they are tight.
\end{openproblem}

\subsection{Complexity of Computing Graph Ranks}

In our opinion, the most striking problems left open by our work are the
question of finding algorithms (and hopefully favorable complexity bounds) to
compute the exact (eventual) pp-(increasing-)rank of an input graph.
Indeed, we were able to show only one complexity lower bound (for the
pp-increasing-rank) and no complexity upper bound at all (not even a
decidability bound).

\begin{openproblem}\label{op:pp-rank-comp}
    Find any algorithm to compute either:
    \begin{enumerate}[a.]
      \item\label{op:pp-rank-comp:classic} the pp-rank of a given graph;
        \item\label{op:pp-rank-comp:incr} the pp-increasing-rank of a given graph; or
        \item\label{op:pp-rank-comp:eve-incr} the eventual pp-increasing-rank of a graph.
    \end{enumerate}
\end{openproblem}

Note that \cref{t:pp-rank,p:peri-comp} provide computable upper bounds for all three ranks.
Hence, \cref{op:pp-rank-comp} hinges on deciding whether, given a graph~$G$ and integer~$k$, the graph~$G$ is (eventually) $k$-orderable (by increasing length).
Recall that we do not know if any graph has a pp-rank of 3 (\cref{op:3-orderable}).  
If there is none, \cref{op:pp-rank-comp:classic} would amount to testing whether a graph pp-rank is 1, something that seems way easier.

\begin{openproblem}
    Show NP-hardness results for the pp-rank of a given graph or the eventual pp-increasing-rank of a given graph.
\end{openproblem}

Computing any of the three ranks essentially amounts to computing a Hamiltonian path in one or several walk graphs.
Hence this task should intuitively be NP-hard.
However, it turns out that it is quite hard to control the shape of a walk graph.
Computing the ranks could even be in higher classes.

\subsection{Odd Ranks}

    It is mentioned in the main text that~$\ppreleq{\ell}$ only makes sense
    if~$\ell$ is an even number. Indeed, we
    have~${\ppreleq{2k}}={\ppreleq{2k+1}}$ for every~$k\in\N$.
    Since we heavily use $\ppreleq{2k}$ to obtain both lower and upper
    approximations of pp-increasing-ranks, it may be tempting to think that the
    pp-increasing-rank of any graph is always an even number.  
    \cref{ex:odd-incr-rank}, below, shows that it is not always the case.
        
    \begin{figure}[h!]
        \centering
        \begin{tikzpicture}
    \newcommand{\dist}{10mm}\input{figure/graphstyle}

    \path 
            node[vertex, finitary] (v1) {} 
            ++(0:\dist) 
                node[vertex, finitary] (v2){}
            ++(60:\dist) 
                node[vertex, finitary] (v4){};
    \path (v2)
            ++(-30:\dist) 
                node[vertex, finitary] (v5){}
            ++(0:\dist) 
                node[vertex, finitary] (v6){};
    \path (v5)
            ++(-120:\dist) 
                node[vertex, finitary] (v3){}
    ;

    \path[edge,finitary] (v1) to node[below,pos=.3]{$a$} (v2);
    \path[edge,finitary] (v2) to node[above left,pos=.3]{$b_1$} (v4);
    \path[edge,finitary] (v2) to node[above,midway]{$b_2$} (v5);
    \path[edge,finitary] (v3) to node[below right,midway]{$b_3$} (v5);
    \path[edge,finitary] (v5) to node[above,pos=.3]{$c$} (v6);

\end{tikzpicture}%
        \caption{A graph with an odd pp-increasing-rank}
        \label{f:odd-incr-rank}
    \end{figure}
        
    \begin{example}\label{ex:odd-incr-rank}
        We consider the graph from \cref{f:odd-incr-rank} and denote it by~$G$.  
It has four walks of length~$2$ ($ab_1$, $ab_2$, $b_2c$ and $b_3c$), a single walk of length~$3$ ($ab_2c$), and no walk of length~$4$ or more. Let us show the pp-increasing-rank of~$G$ is~$3$.

        First note that~$G$ is not $1$-orderable by increasing length: $|\walks[2]{G}|>1$ and~$\pprel{1}$ never relates two different walks of the same length.

        Next, let us show that $G$ is not $2$-orderable by increasing length.
        For this, observe that~$(\walks[2]{G},\pprel{2})$ is an undirected line graph, so that there are only two 2-bounded orderings of~$\walks[2]{G}$, namely $(ab_1,ab_2,b_2 c,b_3 c)$ and its reverse.
If there were an increasing 2-bounded ordering of~$\walks{G}$, then the predecessor of~$ab_2c$ would necessarily be~$ab_1$ or~$b_3 c$, neither of which is related to~$ab_2c$ by~$\pprel{2}$, a contradiction. Hence, $G$ is not $2$-orderable by increasing length.

        On the other hand, since~$(\walks[1]{G},\pprel{3})$ is a clique and~$ab_2c \pprel{3} b_3c$, it is clear that~$G$ is $3$-orderable by increasing length.
        Hence, the pp-increasing-rank of~$G$ is~$3$.
    \end{example}

Indeed \cref{ex:odd-incr-rank} gives a case where the hardest part of the ordering is not to find an ordering for all walks of a given length, but to connect those orderings.
We don't know whether such example exists for the eventual pp-increasing rank:

\begin{openproblem}
    Find a graph with an eventual pp-increasing-rank that is an odd number, or show it is impossible.
\end{openproblem}

\end{toappendix}

\section{Preliminaries}
\label{s:preliminaries}

We first recall useful graph definitions and fix notations. Then we introduce
the push-pop distance on walks and define our orderability notions on the
walks of a graph.

\subsection{Classical Graph Definitions}

For every integer~$n$, we write $[n]=\{0,\ldots,n-1\}$. By \defemph{initial segment of~$\N$} we refer to the whole set~$\N$ or any set~$[n]$ for~$n\in\N$.  An \defemph{ordering} of a countable set~$S$ is a bijection $i \mapsto s_i$ from an initial segment~$I$ of~$\mathbb{N}$ to~$S$, written $s_0, s_1, \ldots$ or~$(s_i)_{i\in I}$; 
it may be finite or infinite depending on the finiteness of~$I$.
For~$E\subseteq X\times X$ a binary relation over a set~$X$, we use the notation~$E^{n}$, $E^{\leq n}$, and~$E^*$ as usual:
$E^0$ is the identity over~$X$, for $n\geq 0$ the relation  $E^{n+1}$ is  $\{ (x,z) \mid \exists y, x \mathrel{E^{n}} y \text{ and } y \mathrel{E} z \}$, the relation
    $E^{\leq n}$ is  $\bigcup_{i=0}^{n} E^i$, and $E^*$ is $\bigcup_{i=0}^{\infty} E^i$.

\subparagraph*{Graphs and walks.}
A directed \defemph{graph} (or simply \emph{graph})
is a pair $G = (V, E)$ where $V$ is the set of \defemph{vertices}, and $E\subseteq V\times V$ is the set of \defemph{edges}.
All graphs are implicitly finite unless stated otherwise.
For every edge~$e=(s,t)\in E$, we call $s$, resp.~$t$ and denote by~$\src(e)$, resp.~$\tgt(e)$ the \defemph{source} of~$e$, resp.~the \defemph{target} of~$e$.
Note that we allow \defemph{self-loops} of the form $(v,v)$ in~$E$.
We collectively call the elements of~$V\cup E$ the \defemph{objects} of~$G$.
A \defemph{subgraph} of $G$ is a graph~$(V',E')$ where $V'\subseteq V$ and $E'\subseteq E$.

A \defemph{walk} in a graph $G = (V, E)$ is a finite sequence $w=(e_1, \ldots, e_k)$ of edges in~$E$ with $\tgt(e_i) = \src(e_{i+1})$ for each $1 \leq i < k$; we call~$k$ the \defemph{length} of~$w$ and denote it by~$|w|$.
When~$w$ is nonempty, we also say that $w$ is a walk from $\src(w) \coloneq \src(e_1)$ to $\tgt(w) \coloneq \tgt(e_k)$.
The empty sequence of edges is called the \defemph{empty walk} and denoted by~$\varepsilon$.
For $u$ and $v$ two walks, if at least one of them is empty or if $\tgt(u) = \src(v)$, we write $u \cdot v$ or simply $uv$ for the \defemph{concatenation} of~$u$ and~$v$.
The set of all walks in~$G$ is denoted by~$\walks{G}$; we also use the notation~$\walks[n]{G}$, resp.~$\walks[\leq n]{G}$, resp.~$\walks[\geq n]{G}$ to denote the walks in~$G$ of length~$n$, resp.~at most~$n$,
resp.~at least~$n$. 

\begin{remark}\label{r:multi-graph}
We define walks as \emph{edge} sequences rather than \emph{vertex} sequences, motivated by our future application to databases.
Although we restrict our analysis to 
graphs for the sake of clarity, this convention ensures our results generalize immediately to 
multigraphs, the predominant model for graph databases.
\end{remark}

\subparagraph*{Strong connectedness.}

Two vertices $v_1$ and $v_2$ of a graph~$G$ are \defemph{strongly connected} if either $v_1=v_2$ or 
there is a nonempty walk from~$v_1$ to~$v_2$ and a nonempty walk from~$v_2$ to~$v_1$. 
Strong connectedness is an equivalence relation, whose classes are the
\defemph{strongly connected components} (\defemph{scc} for short) of~$G$. A graph is strongly connected if it has 
at most one scc. An scc
is \defemph{trivial} if it consists of a single vertex with no self-loop. A
graph is \defemph{acyclic}
if all its sccs are trivial, and \defemph{cyclic}
otherwise.
The \defemph{condensation}~$\mathcal{C}$ of~$G$~\cite{harary1969} is the 
acyclic graph
with the sccs
of~$G$ as vertices, and with an edge from~$C_1$ to~$C_2$ whenever $C_1 \neq C_2$ and~$G$ has an edge
from some $v_1 \in C_1$ to some $v_2 \in C_2$.

We say that a nonempty walk \defemph{starts in} (resp.~\defemph{ends in},
resp.~\defemph{encounters})
    an scc~$C$ if its first vertex (resp.~its last vertex, resp.~some of its vertices) belongs to~$C$.
We also say that a walk \defemph{entirely belongs}
to a nontrivial scc~$C$ if all its vertices belong to~$C$.

\subparagraph*{Weak connectedness and undirected graphs.}
We sometimes consider graphs while ignoring the orientation of their edges.
An \defemph{undirected graph} is a graph $U=(V,E)$ such that $(s,t) \in E$
implies $(t,s) \in E$ for all $s,t \in V$. 
Given a directed graph $G=(V,E)$, let $G'$ be the undirected graph $G' = (V, E \cup \{(t,s) \mid (s,t) \in E\})$. Two vertices $v,v'$
of~$G$ are \defemph{weakly connected} if they are strongly connected in~$G'$ or equivalently if either $v=v'$ or there is a walk in~$G'$
from $v$ to $v'$ or from $v'$ to $v$. We say that $G$ is \defemph{weakly connected} if $G'$ is strongly connected. Note that weak and
strong connectedness coincide for undirected graphs; we then simply say \defemph{connected}. For
clarity, 
a walk in an undirected graph is called a \defemph{chain}.

\subsection{Graph Orderability}

We are interested in the walks that can be transformed into one another with at most~$k$ push and pop operations at either end. We materialize this notion as the relation~$\pprel{k}$, defined next.
First, two walks~$w,w'$ satisfy $w \pprel{1} w'$ iff either $w=w'$, $w\cdot e = w'$ ($w'$ is obtained from~$w$ with a \emph{push-right} operation), $e\cdot w = w'$ (resp.~~\emph{push-left}), $w = w'\cdot e$ (resp.~\emph{pop-right}), $w = e\cdot w'$ (resp.~\emph{pop-left}), for some edge~$e$.
Second, $\pprel{k}$  is defined as~$\pprel{1}^{\leq k}$.
A similar notion, defined as a distance function, is used on words in~\cite{amarilli2023enumerating} and was introduced earlier under the name \emph{substring distance} in~\cite{ng2017state}.
We will also use~$\ppreleq{k}$, a refinement of~$\pprel{k}$ that only relates
walks of the same length: $u\ppreleq{k} v \iff  u\pprel{k}v \wedge |u|=|v|$.

\begin{remark}\label{r:ppreleq-expl}
    The $\ppreleq{k}$ relation is less intuitive than it seems.
First, $\ppreleq{2k+1}$ is equal to $\ppreleq{2k}$, for all~$k\in\N$.
Indeed, a single push or pop operation always changes the length of the walk, hence the extra operation is useless in that case.
Second, $\ppreleq{k}^{\leq j}$ is always included in $\ppreleq{kj}$, 
    but the converse usually fails.
Indeed, $\ppreleq{k}$ only relates walks of the same length, hence $\ppreleq{k}^{\leq j}$ requires the (up to)~$j-1$ intermediate walks to \emph{all have the same length}, while $\ppreleq{kj}$ imposes no such constraints on intermediate walks.
For instance, consider the disjoint union of two walks~$w_1,w_2$ of length~$k$.
Notice that~$w_1$ and~$w_2$ are related by~$\ppreleq{2k}$ (pop all edges from one walk and push all edges from the other) but not by~$\ppreleq{j}^{*}$, for any~$j<2k$.
\end{remark}

A \defemph{walk ordering} of a graph~$G$
is an ordering of~$\walks{G}$.
We are interested in two properties of a walk ordering
$(w_i)_{i\in I}$, both quantified over the indices~$i$ such that $i$ and~$i+1$
are in~$I$:
it is \defemph{$k$-bounded} if $w_i\pprel{k} w_{i+1}$ for all~$i$, and
\defemph{increasing} if $|w_i| \leq |w_{i+1}|$ for all~$i$. Note that the
increasing condition does not require the lengths to be \emph{strictly} increasing, as a
graph may have many walks of the same length (and the order between them is unconstrained).

\begin{definition}[Graph orderability]\label{def:orderable}
  Let~$G$ be a graph and let $k \in \mathbb{N}$. We say that~$G$ is:
  \begin{enumerate}[a.]
    \item \label{d:k-orderable}\defemph{$k$-orderable} if $\walks{G}$ admits a $k$-bounded ordering;
    \item \label{d:eve-k-orderable}\defemph{eventually $k$-orderable} if $\walks[\geq \ell]{G}$ admits a
      $k$-bounded ordering for some $\ell \in \mathbb{N}$; 
    \item \label{d:eve-k-orderable-incr}\defemph{$k$-orderable by increasing length}, resp.\
      \defemph{eventually $k$-orderable by increasing length}, if the ordering
      in \cref{d:k-orderable}, resp.~in \cref{d:eve-k-orderable}, is moreover
      required to be increasing;
    \item \defemph{orderable} if it is $k$-orderable for some~$k$; we define
      \defemph{eventually orderable}, \defemph{orderable by increasing length}
      and \defemph{eventually orderable by increasing length} similarly.   \end{enumerate}
\end{definition}

We point out that orderable and eventually orderable are the exact same notion, and the same holds for orderable by increasing length and eventually orderable by increasing length; but the distance bounds that they achieve can be distinct.
Furthermore, note that a $k$-bounded ordering is $k'$-bounded for every $k' \geq k$, so the set of
integers~$k$ witnessing each of the notions above is upward closed.
Moreover, the following minima exist:

\begin{definition}\label{def:pp-rank}
  The \defemph{pp-rank} of an orderable graph~$G$ is the least~$k$ such that~$G$
  is $k$-orderable. We define \defemph{pp-increasing-rank} and \defemph{eventual
  pp-increasing-rank} analogously.
\end{definition}

\section{Inner-Connectedness and Terminology for Vertices, Edges and Walks }
\label{s:typology}

\begin{toappendix}
    \label{a:typology}
\end{toappendix}

In this section we give a necessary structural condition for graph orderability,
called \emph{inner-connectedness}. It specializes the notion of
\emph{interchangeability} used in~\cite{amarilli2023enumerating} for automata.
To define it, we assign a \emph{kind} to
the vertices and edges of a graph according to how they are connected to
nontrivial sccs. We then define inner-connectedness and show that
inner-disconnected graphs are not orderable.
We close the section by introducing the notion of \emph{walk types}, which we will use
throughout the following sections.

\subsection{Object Kinds and Inner-Connectedness}
\begin{toappendix}
    \subsection{Object Kinds and Inner-Connectedness}
\end{toappendix}
Let~$G$ be a graph, and recall that an \emph{object} of~$G$ is a
vertex or an edge. We consider the two following properties that an object~$o$ of~$G$ may satisfy.
\begin{subequations}\label{eq:kind}
\begin{gather}
\label{eq:kind:after-scc}
\text{$o$ is part of a walk that \textbf{starts} in some nontrivial scc} \\
\label{eq:kind:before-scc}
\text{$o$ is part of a walk that \textbf{ends} in some nontrivial scc}
\end{gather}
\end{subequations}

\noindent These two properties partition the objects of~$G$ into four
\defemph{kinds}. An object is:
\begin{itemize}
  \item \defemph{inner}, or \emph{of kind~\defemph{\I}}, and coloured \textcolor{innercolor}{\textbf{blue}} in figures, if it satisfies both
    \eqref{eq:kind:after-scc} and \eqref{eq:kind:before-scc};
  \item \defemph{left-outer}, or \emph{of kind~\defemph{\L}}, and coloured \textcolor{outercolor}{\textbf{orange}} in figures, if it satisfies
    \eqref{eq:kind:before-scc} only;
  \item \defemph{right-outer}, or \emph{of kind~\defemph{\R}}, and coloured \textcolor{outercolor}{\textbf{orange}} in figures, if it satisfies
    \eqref{eq:kind:after-scc} only;
  \item \defemph{finitary}, or \emph{of kind~\defemph{\F}}, and coloured
    \textcolor{finitarycolor}{\textbf{mint green}}, if it satisfies
    neither.
\end{itemize}
Left-outer and right-outer objects are collectively called \defemph{outer}
objects.

Every vertex of a nontrivial scc is an inner object, but the converse fails: an
inner vertex may form a trivial scc on its own, and an inner edge may lie outside
every nontrivial scc, as illustrated in~\cref{f:inner-connected:x}.  We may now define inner-connectedness.

\begin{definition}
  The \defemph{inner subgraph} of~$G$ is the subgraph consisting of its inner
  objects. We say that $G$ is \defemph{inner-connected} if its inner subgraph is
  weakly connected
(in particular if the inner graph is empty), and \defemph{inner-disconnected} otherwise.
\end{definition}

Two inner-disconnected graphs are shown in Figure~\ref{f:inner-disconnected}, while
Figure~\ref{f:inner-connected} shows graphs that are inner-connected.
Finally, Figure~\ref{f:inner-connected:man} illustrates several possible ways outer and finitary edges may combine.
Inner-connectedness coincides with the notion of DFA interchangeability
of~\cite{amarilli2023enumerating} (in the special case where all transitions use distinct letters), so that
\cite[Theorem~3.9]{amarilli2023enumerating} yields the following statement in our
setting.

\begin{toappendix}
  In this appendix section, we explain why the following result is true:
\end{toappendix}

\begin{propositionrep}[\cite{amarilli2023enumerating}]\label{p:inner-disconnected=>nonorderable}
  No inner-disconnected graph is eventually orderable.
\end{propositionrep}

For the sake of completeness, we provide a self-contained proof of \cref{p:inner-disconnected=>nonorderable} in Appendix~\ref{a:typology}.
In the remainder of the paper we focus on inner-connected graphs.
Notice that this includes all acyclic graphs since their inner subgraph is empty, hence weakly connected.

\begin{toappendix}
We first explain why the result follows from \cite{amarilli2023enumerating},
  before giving a self-contained proof which is easier to follow as it only uses
  graph terminology.

  \subparagraph*{Proof of \cref{p:inner-disconnected=>nonorderable} via
  \cite{amarilli2023enumerating}.}
  To obtain the result from \cite{amarilli2023enumerating}, we will rely on the
  following:

  \begin{claim}
    \label{clm:graphtodfa}
    Given a directed graph $G = (V, E)$, we can build in linear time a DFA $A_G$ whose alphabet is
  the edge set $E$ of $G$, such that the language accepted by~$A_G$ is precisely
  the set of walks of~$G$.
  \end{claim}

  To prove \cref{clm:graphtodfa}, intuitively, we see $G$ as an automaton where each
  vertex is a state and each edge $e \in E$ is a transition labeled by the
  letter~$e$, and all states are both initial and final. The only subtlety is
  that this does not give a deterministic automaton (because of the multiple
  initial states), whereas the results of \cite{amarilli2023enumerating} only
  apply to DFAs. To address this, we add one dedicated initial state. Formally:

  \begin{proof}[Proof of \cref{clm:graphtodfa}]
    Let $G = (V, E)$ be the input directed graph. Let $Q \coloneq V \uplus
    \{q_0\}$, where $q_0$ is a fresh state that does not occur in~$V$. Let $A_G$
    be the automaton on alphabet $E$ with state space $Q$, with all states being
    final, and with the following kinds of transitions:
    \begin{itemize}
      \item For each $e \in E$, we have a transition labeled $e$ from $q_0$ to
        $\tgt(e)$;
      \item For each $e \in E$, we have a transition labeled $e$ from $\src(e)$ to
        $\tgt(e)$;
    \end{itemize}
    (Intuitively, we can also obtain $A_G$ by adding only the second kind of
    transitions together with $\epsilon$-transitions from~$q_0$ to every state
    in~$V$ and then doing $\epsilon$-transition removal.)
    It is immediate that $A_G$ as built is indeed a deterministic automaton
    (because for each letter $e \in E$ there are only two transitions
    labeled~$e$ and their source states are different), that the construction is
    in linear time, and that a sequence of edges of~$E$ is accepted by~$A_G$
    precisely if it is a walk of~$G$.
  \end{proof}

  We can now prove \cref{p:inner-disconnected=>nonorderable} by appealing to the
  results of \cite{amarilli2023enumerating}. Let $G$ be an inner-disconnected
  graph. Using \cref{clm:graphtodfa}, let $A_G$ be a DFA recognizing the
  language $L$ of the walks of~$G$. Let us show that $L$ is not eventually
  orderable.

  For this, let us first show that $A_G$ is not interchangeable in the
  sense of~\cite{amarilli2023enumerating}: we show by contraposition that if
  $A_G$ were interchangeable then $G$ would be inner-connected. Indeed, the
  non-trivial sccs of $A_G$ are in bijection with the non-trivial sccs of~$G$,
  and the loopable states of $A_G$ (in the terminology
  of~\cite{amarilli2023enumerating}) are in bijection with the vertices of the
  non-trivial sccs of~$G$.
  Further, if $A_G$ were interchangeable, then all its loopable vertices would
  be related by the transitive closure of the connectedness relation
  of~\cite{amarilli2023enumerating}: it is easy to see from the structure
  of~$A_G$ that the compatibility relation on~$A_G$ is the identity because
  non-empty runs starting at a loopable state $x$ of~$A_G$ must start with a letter $e \in
  E$ with $\src(e) = x$, so there cannot be a non-empty word that labels a loop
  on two different loopable states of~$A_G$. Thus, if $A_G$ were
  interchangeable, then the graph formed of its loopable vertices and the
  directed paths from a loopable vertex to another (possibly via non-loopable
  vertices) would be weakly connected, implying that $G$ would then be
  inner-connected. Hence, by contraposition, as $G$ is not inner-connected then
  $A_G$ is not interchangeable.

  We can then invoke \cite[Theorem 3.9]{amarilli2023enumerating}. This result
  implies that the language of $A_G$ is not orderable for the Levenshtein
  distance, hence also for the push-pop-distance, and so indeed $G$ is not an
  orderable graph. This establishes \cref{p:inner-disconnected=>nonorderable},
  noting that if $G$ were \emph{eventually orderable} with some finite distance
  bound $d$ then it would be \emph{orderable} with some finite distance bound $d'$,
  e.g., we can take $d'$ to be the maximum between $d$ and the (finite) maximal distance between two
  consecutive walks in the (finite) prefix of the witnessing sequence for
  eventual orderability to which the distance bound $d$ does not apply.

  \subparagraph*{Self-contained proof of \cref{p:inner-disconnected=>nonorderable}.}
For our self-contained proof of \cref{p:inner-disconnected=>nonorderable}, we
will rely on the following preliminary lemma, which is the analogue in our
setting of \cite[Proposition 3.10]{amarilli2023enumerating}.

\begin{lemma}\label{l:innder-disconnected=>eventually-disconnected}
    Let~$G$ be 
    an inner-disconnected graph.
    Then, for every integer~$k \geq 0$, there exists~$n$ such that we can partition $\walks[\geq n]{G}$ into two infinite sets~$W_1,W_2$ such that there are no edges of $\pprel{k}$ between $W_1$ and $W_2$.

\end{lemma}
\begin{proof}

    Let~$G=(V,E)$ be a graph.
    Just like strong connectedness, weak connectedness defines an equivalence relation on~$V$ and we call its equivalence classes the \emph{weakly connected components} (or wcc for short) of~$G$.
    Note that each wcc is a union of sccs.

    We denote by~$I$ the inner subgraph of~$G$ and by~$C_1,\ldots,C_m$ the wccs
    of~$I$; as $G$ is inner-disconnected we have $m \geq 2$.

    Let $n \coloneq |E|+k+1$. We associate each walk~$w$ of~$G$ of length at least $n$ to one of~$C_1,\ldots,C_m$, as follows.
    Since~$n>|E|$, some edge of~$G$ occurs twice in~$w$, hence~$w$
    must encounter at least one nontrivial scc of~$G$; this scc is contained in some wcc~$C_i$ of~$I$; and we associate~$w$ with~$C_i$.
    Note that this correctly defines a mapping from walks of~$G$ of length at
    least~$n$ to the wccs of~$I$, because if $w$ encounters more than one
    nontrivial scc of~$G$ then they must all belong to the same wcc of~$I$. 
    This implies that $w$ contains no edge that belong to a wcc~$C_j$ of~$I$
    with $i\neq j$. As we have $|w| \geq |E|+k+1$, we
    know that $w$ contains at least $k+1$ edge occurrences that are repeated,
    and the corresponding edges must belong to a nontrivial scc of~$G$, hence
    to~$C_i$.

    For every~$i\in[m]$, we let~$W_i$ denote the set of the walks in $\walks[\geq n]{G}$ associated with the wcc $C_i$.
    Note that each wcc~$C_i$ contains at least one nontrivial scc (by definition of the inner subgraph), hence each~$W_i$ contains infinitely many walks.
It follows from the previous paragraph that two walks~$w_i\in W_i$
    and~$w_j\in W_j$ are never related by~$\pprel{k}$ unless $i=j$ (because at
    least $k+1$ edge occurrences have to be removed and $k+1$ edge occurrences
    have to be added in order to transform $w_i$ into $w_j$).
Hence, since~$W_1\cup\cdots\cup W_m=\walks[\geq n]{G}$ and~$m\geq 2$, the
    statement follows by partitioning $\walks[\geq n]{G}$ into $W_1$ and $W_2\cup\cdots\cup W_m$.
\end{proof}

We can now prove \cref{p:inner-disconnected=>nonorderable} similarly to the
proof of \cite[Theorem 3.9]{amarilli2023enumerating}:

\begin{proof}[Proof of \cref{p:inner-disconnected=>nonorderable}]
    For the sake of contradiction, let us assume that there is an inner-disconnected graph~$G$ that is eventually~$k$-orderable, yielding an integer~$\ell$ and an (infinite) walk sequence~$w_1,\ldots,w_i,\ldots$ that enumerates~$\walks[\geq\ell]{G}$.
Applying \cref{l:innder-disconnected=>eventually-disconnected}, let $n$ and $W_1,W_2$ be two infinite sets of vertices in~$(\walks[\geq
n]{G},{\pprel{k}})$ that are disconnected under $\pprel{k}$.
    In particular,~$W_1'=W_1\cap\walks[\geq\ell]{G}$ and~$W_2'=W_2\cap\walks[\geq\ell]{G}$ are infinite as well.
    There are finitely many walks in~$\walks[\geq\ell]{G}\cap \walks[<n]{G}$, possibly none, hence there is
    a rank~$r$ such that all walks~$w_i$ with $i\geq r$, have length at least~$n$.
Now, let $X=\{w_i \mid i < r\}$ be the set of walks produced before rank~$r$.
    As $X$ is finite, there are two walks $v_1$ in $W_1'$
    and~$v_2$ in $W_2'$ that do not belong to~$X$.
    Hence, $v_1=w_i$ and~$v_2=w_j$ for some~$i,j\geq r$; without loss of generality we may assume that~$i<j$.
    We reach a contradiction because $w_i,w_{i+1}\ldots,w_j$ provides a chain
    connecting~$w_i$ and~$w_j$, and this chain is in~$(\walks[\geq
    n]{G},\pprel{k})$ from our assumption that all walks enumerated after rank~$r$
    have length at least~$n$.
\end{proof}

\begin{remark}\label{r:levenshtein-inner-connected}
The proof of \cref{l:innder-disconnected=>eventually-disconnected} and thus of
\cref{p:inner-disconnected=>nonorderable} relies on the number of edges which
are different for walks in two different weakly connected components of the
inner subgraphs. Hence, the position of these edges does not change the proof,
which still works for the Levenshtein distance, which is a relaxed version of
the Push-Pop distance.  
(We note that the same holds in \cite{amarilli2023enumerating}, with
\cite[Theorem 3.9]{amarilli2023enumerating} applying also to the Levenshtein
distance.)
See \cref{s:conclusion} where the issue is discussed in
more detail.
\end{remark}

\end{toappendix}

\subsection{Walk Types}
\begin{toappendix}
    \subsection{Walk Types}
\end{toappendix}

Following our definition of object kinds, we now define \emph{walk types}, which we use afterwards to define peripherality and eventual peripherality.
Most proofs will also proceed by case analysis on the walk types.

A nonempty walk consisting only of edges of kind~$t \in \{\L,\I,\F,\R\}$ is
called a walk \defemph{of type~$t$}. Note that its first and last vertices may have a kind
other than~$t$ (e.g., the walk $\b3\b4$ in \cref{f:inner-connected:man} is of type~$\R$ but starts in a vertex of kind~$\I$).

A \defemph{type specification} is a nonempty sequence $\tau$ over $\{\L,\I,\F,\R\}$ in which
no element occurs twice and each element is optionally marked with~${}^?$.
For $\tau = t_0\cdots t_k$,
a walk~$w$ is
\defemph{of type~$\tau$}, or a \defemph{$\tau$-walk}, if we can write $w = w_0 \cdots w_k$, where
$w_i$ is a walk of type~$t_i$ for each $0 \leq i \leq k$, and $w_i$ is moreover
allowed to be empty when $t_i$ carries a~${}^?$.
A walk may be of several types, e.g.,
a walk can be both of type~$\I$ and of type~$\L^?\I$.
For a fixed
specification however, the factorization is unique, since the letters are
pairwise distinct and the kind of an edge of~$w$ determines which factor it
belongs to. We call it \defemph{the factorization of~$w$ with respect
to~$t_0\cdots t_k$}. For $k=0$ this agrees with the definition of the previous
paragraph.

We state a few properties of walk types below; in particular,
\cref{pp:walk-type:supertypes} provides the two main types we use in case analysis.

\begin{propertyrep}Let~$G$ be a graph.\label{pp:walk-type}
    \begin{enumerate}[a.]
        \item\label{pp:walk-type:factor} 
        Let~$w_1\cdot w_2$ be a walk such that~$w_1$ and~$w_2$ are nonempty.
        Then~$w_1\cdot w_2$
        is of type~\L{} (resp.~\R)
        if and only if~$w_2$ is of type \L{} (resp.~$w_1$ is of type~\R).
        
        \item\label{pp:walk-type:no-i&f} 
        No walk contains both an edge of kind~$\I$ and an edge of kind~$\F$.

        \item \label{pp:walk-type:supertypes}
        Each walk in~$G$ is either of type~$\L^?\I^?\R^?$ or of type~$\L^?\F\R^?$
     \end{enumerate}
\end{propertyrep}

\begin{proof}
    \cref{pp:walk-type} is a  consequence of the way walk concatenation combines with the properties in \cref{eq:kind:after-scc,eq:kind:before-scc}.
    It may indeed be verified that, if~$w$ denotes a walk and~$e,e'$ two edges
    such that $e\cdot w\cdot e'$ is a walk, then we have the following two
    implications:
    \begin{enumerate}
        \item $e$ satisfies \eqref{eq:kind:after-scc} $\implies$ $e'$ satisfies \eqref{eq:kind:after-scc}
        \item $e'$ satisfies \eqref{eq:kind:before-scc} $\implies$ $e$ satisfies \eqref{eq:kind:before-scc}
    \end{enumerate}
    \smallskip
    
    Proof of \cref{pp:walk-type:factor}. We focus on the statement about type~$\L$,
    as the statement about type~$\R$ is symmetric.
    The forward direction is immediate.
    For the converse direction, if~$w_2$ is of type~$\L$, then the last edge of~$w_2$ is of kind~$\L$, hence by definition it satisfies \cref{eq:kind:before-scc} but not \cref{eq:kind:after-scc}.
    Implication~2 above implies that every edge of~$w_1\cdot w_2$ satisfies \cref{eq:kind:before-scc},
    and the contraposition of Implication~1 implies that no edge of~$w_1\cdot
    w_2$ satisfies \cref{eq:kind:after-scc}, hence it is indeed the case that
    all edges in~$w_1$ are of kind~$\L$.

    \medskip
    
    Proof of \cref{pp:walk-type:no-i&f}.   A walk containing an $\I$-edge~$e_\I$ and an \F-edge~$e_\F$ is of the form~$w\cdot e_\I \cdot w' \cdot e_\F \cdot w''$ or of the form~$w\cdot e_\F \cdot w' \cdot e_\I \cdot w''$.
    The former case contradicts Implication 1 and the latter case Implication~2.
    
    \medskip

\cref{pp:walk-type:supertypes} is a direct consequence of \cref{pp:walk-type:factor,pp:walk-type:no-i&f}.
\end{proof}

\section{Peripherality and Upper Bounds}
\label{s:peripherality}
\begin{toappendix}
    \label{a:peripherality}
\end{toappendix}

We now have the requisite definitions 
to show the upper bounds for the
different push-pop ranks (\cref{t:pp-rank},
\cref{t:pp-rank:classic,t:pp-rank:incr,t:pp-rank:eve-incr}).
Our strategy generally relies on the 
``even-odd trick''~\cite{sekanina1960ordering,karaganis1968cube,feder1994network} that
builds a Hamiltonian path in the cube~$(V,E^{\leq 3})$ of any connected finite undirected graph~$(V,E)$.
In our settings, we will apply it to \emph{walk graphs}, whose vertices are the
walks to be ordered and where two walks are connected by an edge if they are at small
push-pop distance. Since these graphs are generally infinite,
we must adapt this technique, which we present in
\cref{s:peripherality:even-odd}: namely, \cref{p:layer:even-odd} gives a generic method to build a $3k$-bounded walk ordering by partitioning walks such that each part is connected for the $\ppreleq{k}$ relation.
Then, building on walk types, we define in \cref{s:peripherality:def} two graph parameters, \emph{peripherality} and its variant \emph{eventual peripherality}.
The values of these parameters provide us with a distance bound~$k$ suitable
for \cref{p:layer:even-odd},
as shown in \cref{s:peripherality:zero} when peripherality is equal to~$0$,
and in \cref{s:peripherality:upper} in the general case.

\subsection{Orderability and Walk Graph Connectedness}
\label{s:peripherality:even-odd}
\begin{toappendix}
    \subsection{Orderability and Walk Graph Connectedness}
\end{toappendix}

In the following, we will often consider possibly infinite undirected graphs
whose vertices are walks (e.g., $\walks[\ell]{G}$ or $\walks[\geq \ell]{G}$) and
whose edge relation is~$\pprel{k}$ or~$\ppreleq{k}$ for some~$k$; we informally
call them \emph{walk graphs}.
We now show how walk orderability of a graph $G$ can be studied through the simpler lens of connectedness of its walk graphs. Specifically, we demonstrate the approach by showing the following result:

\begin{propositionrep} \label{p:incr-len-orderability<=layer-orderability}
    Let~$G$ be a graph and~$k$ a positive integer.    
     If the undirected graph  ${(\walks[\ell]{G}, \ppreleq{k})}$ 
      is connected for every~$\ell$, then~$G$ is $3k$-orderable by increasing length.
\end{propositionrep}

We will leverage this result later to show our upper bounds on orderability by increasing length, with a choice of~$k$ given by our subsequent definitions of peripherality.

The proof of \cref{p:incr-len-orderability<=layer-orderability}  uses the classical ``even-odd
trick''~\cite{sekanina1960ordering,karaganis1968cube,feder1994network}:

\begin{theorem}[\cite{karaganis1968cube}, see also \cite{sekanina1960ordering,feder1994network}]
\label{t:even-odd}
  For every connected undirected (finite) graph~$U=(V,E)$ and distinct
  endpoint vertices ${w,w'\in V}$, there exists an ordering $v_0,v_1,\cdots,v_{n}$ of $V$, with $w=v_0$ and $v_n=w'$ such that $(v_i,v_{i+1})\in E^{\leq3}$, for every $i\in[n]$.
\end{theorem}

Note that 
\cref{t:even-odd}
only applies to \emph{finite}~$U$, and not, e.g., 
to the bi-infinite chain ($V=\mathbb{Z}$ and $E=\{(i,i+1),(i+1,i)\mid i\in\mathbb{Z}\}$).
As walk graphs are generally infinite, to show
\cref{p:incr-len-orderability<=layer-orderability} we need a 
result that applies to possibly \textbf{infinite} undirected graphs:

\begin{toappendix}
  To show \cref{p:incr-len-orderability<=layer-orderability}, we first restate and show our auxiliary result:
\end{toappendix}

\begin{propositionrep}\label{p:layer:even-odd}
    Let~$U=(V,E)$ be a possibly infinite undirected graph, and let $(V_i)_{i\in I}$ be  a partition of $V$ as non-empty finite vertex sets meeting the following conditions.
    \begin{enumerate}[a.]   
\item \label{p:layer:even-odd:pre:subgraph-connected} The subgraph of~$U$ induced by~$V_i$ is connected, for every~$i\in I$.
        \item \label{p:layer:even-odd:pre:direct-seq} There is a vertex sequence~$(w_i)_{i\in I}$ such that:
        \begin{itemize}
            \item~$w_i\in V_i$ for all~$i\in I$
            \item~$(w_{i-1},w_i)\in E^{\leq2}$ for all~$i\in I \setminus\{0\}$.
        \end{itemize}
    \end{enumerate}
    Then, there exists an ordering~$(v_j)_{j\in J}$ of~$V$ meeting the following conditions.
\begin{enumerate}[a.]\setcounter{enumi}{2}
    \item \label{p:layer:even-odd:post:incr}
    For every~$i<i'$, whenever we have 
    $v_j\in V_i$ and~$v_{j'}\in V_{i'}$ for some $j,j'$, then $j<j'$.
    \item \label{p:layer:even-odd:post:chain}
    We have~$(v_{j-1},v_{j})\in E^{\leq 3}$, for every $j\in J\setminus \{0\}$.
\end{enumerate}
\end{propositionrep}

The simplest choice for the $V_i$ is to take the walks of length~$i$, which we will do to show
\cref{p:incr-len-orderability<=layer-orderability}. 
However, we have phrased  \cref{p:layer:even-odd} in a more generic way because we will later use it with sets $V_i$ defined according to the peripherality notion that we will introduce. We also note that 
\cite[Lemma~4.6]{amarilli2023enumerating} uses a similar result to \cref{p:layer:even-odd} 
(under the name of \emph{stratum-connectedness}); but in their work the sets $V_i$ always correspond to walks of a specific length \textbf{range} (called a \emph{stratum}), which does not suffice for our purposes.

\begin{proof}
    For every~$i\in I$, we define~$w_i'$ from~$w_i$, given by hypothesis~\ref{p:layer:even-odd:pre:direct-seq}, as follows.
If~$|V_i|=1$, that is if~$V_i=\{w_i\}$, we fix~$w_i'=w_i$, otherwise we have~$|V_i|\geq 2$ and we fix~$w_i'$ as any vertex in~$V_i$ such that~$(w_i,w_i')\in E$ and~$w_i\neq w_i'$, which must exist as a consequence of hypothesis~\ref{p:layer:even-odd:pre:subgraph-connected}.
    Note that in both cases~$(w_i,w_i') \in  E^{\leq 1}$.
    
    For every~$i\in I$, we now do one of two things depending on $|V_i|$. The first case is when~$|V_i|>1$. In this case, we apply \cref{t:even-odd} to the subgraph induced by $V_i$, with endpoint vertices $w_i,w_i'$.
    \cref{t:even-odd}~yields, for each integer~$i\in I$ such that~$|V_i|>1$, an ordering $(w_i=v_0^{(i)},v_1^{(i)},\ldots, v_{n_i}^{(i)}=w_i')$ of~$V_i$ such that $(v_j^{(i)},v_{j+1}^{(i)})\in E^{\leq3}$ for each $0 \leq j < n_i$.
    The second case is when $|V_i|=1$. In this case, we set~$n_i\coloneq 0$ and associate with~$i$ the 1-element sequence~$(w_i=v_0^{(i)}=v_{n_i}^{(i)}=w_i')$. Note that, in both cases, we obtain an ordering of~$V_i$ that starts at $w_i$ and ends at $w_i'$ and where consecutive pairs are in $E^{\leq 3}$.
    
    We then concatenate the vertex sequences associated with each~$i$ by increasing~$i$, yielding the following infinite sequence.
    \begin{equation*}
        \overbrace{\underbrace{v_0^{(0)}}_{=w_0},v_1^{(0)},\ldots, \underbrace{v_{n_0}^{(0)}}_{=w_0'}}^{V_0}~,~~
        \overbrace{\underbrace{v_0^{(1)}}_{=w_1},v_1^{(1)},\ldots, \underbrace{v_{n_1}^{(1)}}_{=w_1'}}^{V_1}~,~~
        \ldots~,~~
        \overbrace{\underbrace{v_0^{(i)}}_{=w_i},v_1^{(i)},\ldots, \underbrace{v_{n_i}^{(i)}}_{=w_i'}}^{V_i}~,~~
        \ldots
    \end{equation*}
    It is an ordering of~$V$ since $(V_i)_{i\in I}$ is a partition of $V$ and because $v_0^{(i)},v_1^{(i)},\ldots, v_{n_i}^{(i)}$ is an ordering of $V_i$, for each~$i\in I$.
    Item~\ref{p:layer:even-odd:post:incr} holds by construction.
    Finally, to show item \ref{p:layer:even-odd:post:chain}, we already know that the condition holds within the vertex sequences for each $V_i$, and we must only show it between the last vertex of each sequence and the first vertex of the next sequence. In other words,
we must show that for every $j\in J\setminus \{0\}$ there is a chain of length at most 3 linking $w_{j-1}'=v^{(j-1)}_{n_{j-1}}$ to $v^{(j)}_{0}=w_j$.
    This fact follows from the second item of hypothesis \ref{p:layer:even-odd:pre:direct-seq} ($w_{j-1}$ and $w_{j}$ are at distance~2 or less) and 
    the definition of~$w_{j-1}'$ (either~$w_{j-1}'=w_{j-1}$ or~($w_{j-1}',w_{j-1})\in E$).
\end{proof}

\begin{proof}[Proof sketch of \cref{p:layer:even-odd}]
    Apply \cref{t:even-odd} to each subgraph of~$U$ induced by $V_i$
    with~$w=w_i$ and~$w'$ being any vertex in~$V_i$ such that~$(w_{i},w')\in E$.
    Then, concatenate the obtained orderings.
\end{proof}

\cref{p:incr-len-orderability<=layer-orderability} is then easily obtained 
from \cref{p:layer:even-odd}. 
If~$G$ is cyclic, we use~$V_i=\walks[i]{G}$, for~$i\in\N$, and~$(w_i)_{i\in\N}$ can be set as any walk sequence such that~$|w_i|=i$ and~$w_{i}$ is a prefix of~$w_{i+1}$ (e.g., following any cycle in any nontrivial scc). 
If~$G$ is acyclic, the proof is the same, except we replace~$\N$ by $[n+1]$ where~$n$ is the length of the longest walk in~$G$.

\begin{toappendix}
  We are now ready to show \cref{p:incr-len-orderability<=layer-orderability}:
    \begin{proof}[Proof of \cref{p:incr-len-orderability<=layer-orderability}]
    Let~$G$ be a graph and~$k$ an integer such that the undirected graph~$(\walks[\ell]{G},\pprel{k})$ is connected for every~$\ell$.
    We fix~$V_\ell=\walks[\ell]{G}$, for every~$\ell\in\N$.
    
    We first treat the case where~$G$ is cyclic, in which case there exists a walk sequence~$(w_i)_{i\in\N}$ such that~$w_0=\varepsilon$ and~$w_{i+1}=w_i\cdot e_i$, for each~$i\in\N$ (e.g., following any cycle in any nontrivial scc).
    We apply \cref{p:layer:even-odd} to the (infinite) undirected graph~$(\walks{G},\pprel{k})$,
    and sequences~$(V_i)_{i\in\N}$,~$(w_i)_{i\in\N}$.
    By construction, $(V_i)_{i\in\N}$ is a partition of $\walks{G}$, and Item~\ref{p:layer:even-odd}\ref{p:layer:even-odd:pre:subgraph-connected} holds by hypothesis.
    For every~$i\in\N$, we have~$w_{i+1}\pprel{1}w_{i}$ by construction, hence also~$w_{i+1}\left(\pprel{k}\right)^{\leq 2}w_{i}$; in other words Item~\ref{p:layer:even-odd}\ref{p:layer:even-odd:pre:direct-seq} holds.
    \cref{p:layer:even-odd} yields the $3k$-bounded ordering of~$\walks{G}$ since $(\pprel{k})^{\leq3}$ is equal to~$\pprel{3k}$.
    
    The case where~$G$ is acyclic is similar.  Let~$n$ be the length of the longest walk in~$G$, $I=[n+1]$. Hence in particular,~$\biguplus_{i\in I} V_i = \walks{G}$.
    For every~$i\leq n$, let $w_i$ be the prefix of length~$i$ of an arbitrary walk of length~$n$ in~$G$. Apply \cref{p:layer:even-odd} to the (finite) undirected graph~$(\walks{G},\pprel{k})$,
    and sequences~$(V_i)_{i\in I}$,~$(w_i)_{i\in I}$.
    \end{proof}
\end{toappendix}

\subsection{Peripherality}
\label{s:peripherality:def}

We now define the \emph{peripherality} and \emph{eventual peripherality} of
a graph, by first defining the \emph{peripherality} of a walk. The
peripherality of a walk is equal to the number of edges that must be popped 
to obtain a walk that can reach or be reached from a nontrivial scc. Formally:

\begin{definition}\label{d:peripherality}
    Let $w$ be a walk. The \defemph{peripherality of $w$}, denoted $\peri(w)$, is the length of the shortest walk $u$ such that $w = u\cdot v$ and $v$ is of type $\I^?\R^?$, or $w = v\cdot u$ and $v$ is of type $\L^?\I^?$.
We then define the \defemph{peripherality}~$\peri(G)$ of a graph~$G$
    and the \defemph{eventual peripherality}~$\evperi(G)$ \defemph{of a cyclic graph}~$G$ as follows.
    \begin{align}
    \label{eq:peri}
        \peri(G)={}&\max\{\peri(w)\mid w\in\walks{G}\}\\
        \evperi(G)={}&\lim_{\ell\to\infty}\max\{\peri(w)\mid w\in\walks[\geq\ell]{G}\}
        \label{eq:evperi}
    \end{align}
    We conventionally define the \defemph{eventual peripherality} of an acyclic graph to be equal to~$0$.
\end{definition}

As a direct consequence of \cref{pp:walk-type}, 
if we denote by $|w|_X$ the number of edges of kind $X$ in $w$, the peripherality of a walk can be computed as in \cref{p:peri-direct:classic}, below.
Moreover, 
it is easy to show that eventual peripherality can be computed as in \cref{p:peri-direct:eve}.
Together, they imply that peripherality and eventual peripherality are always well defined.

\begin{propertyrep} \label{p:peri-direct}
\begin{enumerate}[a.]
    \item \label{p:peri-direct:classic} For every walk~$w$, $\peri(w) = |w|_\F + \min\big (|w|_\L, |w|_\R\big)$.
    \item \label{p:peri-direct:eve} For every cyclic graph~$G$, $\evperi(G) = \max\big\{\peri(w)\;\big|\; w\in W\big\}$, where~$W$ is the subset of~$\walks{G}$ of those walks that encounter some nontrivial scc.
\end{enumerate}
\end{propertyrep}

\begin{proof}
  Item~\ref{p:peri-direct:classic} is a direct consequence of
  \cref{pp:walk-type}\ref{pp:walk-type:supertypes}, so we focus on
  item~\ref{p:peri-direct:eve}, that is on showing the following equality (where $W$ is defined in the statement).
  \begin{equation}
    \underbrace{\lim_{\ell\to\infty}\max\{\peri(w)\mid w\in\walks[\geq\ell]{G}\}}_{p_1} = \underbrace{\max\big\{\peri(w)\;\big|\; w\in W\big\}}_{p_2}
  \end{equation}
  
  ($p_1$ and $p_2$ are well defined) No walk in the graph can repeat a $\L$-edge, $\R$-edge nor an~$\F$-edge (a repeated edge is always of kind $\I$).
  Hence the peripherality of a walk is at most~$|F|+\min(|R|,|L|)$ where~$L$
  (resp.~$R$, resp~$F$) is the set of $\L$-edges (resp.~$\R$-edges, resp.~$\F$-edges) in the graph. 
  
  ($p_1\leq p_2$)  By definition, there exists a walk~$w$ such that~$\peri(w)=p_1$ and~$w>|V|$. In particular,~$w$ encounter a nontrivial scc hence~$w\in W$, yielding~$p_1\leq p_2$.
  
  ($p_1 \geq p_2$) Whenever
  a walk $w$ encounters a nontrivial scc and satisfy~$|w|=p_2$, then some vertex $v$ of a nontrivial scc must
  occur in~$w$.
  Thanks to the definition of a nontrivial scc, we can then obtain arbitrarily large walks from~$w$
  by inserting arbitrarily long $\I$-walks from~$v$ to itself. By
  item~\ref{p:peri-direct:classic}, the resulting walks have the same
  peripherality, so indeed the peripherality of~$w$ is reflected in~$p_1$ hence satisfying~$p_2=  |w| \leq p_1$.
\end{proof}

For instance, the graph in \cref{f:inner-connected:man} has peripherality~$2$,
attained by the walks~$\b6\b5\b4$ and~$\b6\b5$.
None of them encounters a nontrivial scc, so the eventual
peripherality of this graph is~$1$, witnessed, e.g., by the walks~$\b1\b2$ or $\b1\b3\b4$.
The graph in \cref{f:inner-disconnected2} has peripherality~$3$,
attained by~$\b1\b2\b3\b4$,
but has eventual peripherality~$0$.
The graph in \cref{f:inner-connected:x} 
has peripherality~$1$ and eventual
peripherality~$0$, and all the remaining graphs
in~\cref{f:inner-disconnected,f:inner-connected} have peripherality~$0$.
Note also that the peripherality of an acyclic graph is the length of its longest walk.

Finally, we illustrate in \cref{f:high-peripherality} some families of cyclic graphs whose peripherality is linear in the size
of the graph.
The peripherality of the graph in \cref{f:high-peripherality:finitary}
is attained by the bottom finitary walk; its eventual peripherality is~$0$.
The peripherality of the graph in \cref{f:high-peripherality:eventual}
is $k$ and it is attained by a walk that encounters the nontrivial scc (the right vertex with a
self-loop), hence this graph also has an eventual peripherality of~$k$, i.e., linear in the size of the graph.
The peripherality of the graph in \cref{f:high-peripherality:no-finitary} is attained only by one walk; although this walk contains an inner vertex, it does \textbf{not} encounter a nontrivial scc; hence the eventual peripherality of this graph is in fact~$0$. 
These families showcase the worst cases for peripherality and eventual peripherality: by the pigeonhole principle, it is easily seen that the following inequalities hold for every graph~$G=(V,E)$:
\begin{lemmarep}\label{l:peri-max}
We have: $\peri(G) < |V|$ and $\evperi(G) < \frac{|V|}{2}$.
\end{lemmarep}

\begin{proof}
We distinguish between walks of 
type $\L^?\F\R^?$ and of
type $\L^?\I^?\R^?$.

For a walk $w$ of type $\L^?\F\R^?$, we show an upper bound of $\peri(w) < |V|$. Indeed, any walk of length $\geq |V|$ must contain a repeated vertex, hence it must contain an $\I$-edge. Thus, walks $w$ of type $\L^?\F\R^?$ have length $< |V|$, and since we have $\peri(w) \leq |w|$, we conclude.

For a walk $w$ of type $\L^?\I^?\R^?$, we show an upper bound of $\peri(w) < |V|/2$.
Write the $\L^?\I^?\R^?$-factorization of~$w$ 
as $w = w_\ell w_i w_r$. By 
\cref{p:peri-direct}\ref{p:peri-direct:classic}, we have $\peri(w) = \min(|w_\ell|, |w_r|)$. Hence, assume by contradiction that $\peri(w) \geq |V|/2$, then we have $|w_\ell|, |w_r| \geq |V|/2$.
In particular~$w_\ell$ and~$w_r$ collectively have~$|V|+2$ vertex occurrences.
By the pigeonhole principle, 3 of these occurrences must be part of nontrivial sccs, hence of kind~$\I$.  Since a walk between $\I$-vertices is of type $\I$, then either~$w_\ell$ or~$w_r$ contains an $\I$-edge.

The statement then follows by definition of $\peri(G)$ (\cref{eq:peri} in \cref{d:peripherality}) and by 
\cref{p:peri-direct}\ref{p:peri-direct:eve} for $\evperi(G)$ using the fact that walks of type $\L^?\F\R^?$ cannot encounter a non-trivial scc. 
This last point is because the only vertices in a non-trivial scc that could occur in such a walk would be at the end of the $\L^?$ factor or at the beginning of the $\R^?$ factor, but then this would contradict the fact that the middle edges are of type~$\F$.
\end{proof}
Since by definition, $\evperi(G)\leq \peri(G)$, we have just shown \cref{t:pp-rank:param} of \cref{t:pp-rank}.

\begin{figure}
\subfloat[]{\label{f:high-peripherality:finitary}\hspace*{3.5ex}\begin{tikzpicture}
    \newcommand{\dist}{10mm}
    \input{figure/graphstyle}
    
    \path node[vertex,inner] (v1) {} ++(0:4.25*\dist) 
          node[vertex,inner] (vk) {};
    \path (v1) ++(-90:\dist) node[vertex,outer] (v1p){};
    \path (vk) ++(-90:\dist) node[vertex,outer] (vkp){};
    \path (v1p) ++(0:\dist) node[vertex,finitary] (v2p) {};
    \path (v2p) ++(0:\dist) node[vertex,opacity=0] (v3p) {};
    \path (vkp) ++(180:\dist) node[vertex,opacity=0] (vkmop) {};

    \path[edge,inner] (vk) to (v1);
    \begin{pgfinterruptboundingbox}
        \path[edge,inner,north loop] (v1) to (v1);
        \path[edge,inner,north loop] (vk) to (vk);
    \end{pgfinterruptboundingbox}
    \path[edge,outer] (v1p) to (v1);
    \path[edge,outer] (vk) to (vkp);
    \path[edge,finitary] (v1p) to (v2p);
    \path[edge,finitary] (v2p) to (v3p);
    \path[edge,finitary] (vkmop) to (vkp);
    \path (v3p) to node[midway,finitary] {$\cdots$} (vkmop);
\end{tikzpicture}\vspace{-1.6em}}
\hspace*{11mm}\subfloat[]{\label{f:high-peripherality:eventual}\hspace*{4ex}\begin{tikzpicture}
    \newcommand{\dist}{10mm}\input{figure/graphstyle}

    \path node[vertex,outer] (v1) {} 
            ++(0:\dist) node [vertex, outer] (v2){}
            ++(0:\dist) node [vertex, outer,opacity=0] (v3){}
            ++(0:1.25*\dist) node [vertex, outer,opacity=0] (vkmo){}
            ++(0:\dist) node [vertex, outer] (vk) {}
            ++(-30:\dist) node [vertex, inner] (r) {}
            ++(-150:\dist) node [vertex, outer] (v1p) {}
            ++(-180:\dist) node [vertex, outer] (v2p) {}
            ++(-180:\dist) node [vertex, outer,opacity=0] (v3p) {}
            ++(-180:1.25*\dist) node [vertex, outer,opacity=0] (vkmop) {}
            ++(-180:\dist) node [vertex, outer] (vkp) {}
    ;

    \path[edge,outer] (v1) to (v2);
    
    \path[edge,outer] (v2) to (v3);
    \path (v3) to node[midway,outer] {$\cdots$} (vkmo);
    \path[edge,outer] (vkmo) to (vk);
    \path[edge,outer] (vk) to (r);
    \path[edge,outer] (r) to (v1p);
    \path[edge,outer] (v1p) to (v2p);
    \path[edge,outer] (v2p) to (v3p);
    \path (v3p) to node[midway,outer] {$\cdots$} (vkmop);
    \path[edge,outer] (vkmop) to (vkp);

    \begin{pgfinterruptboundingbox}
        \path (v1.north) node[xshift=-2mm,anchor=south] {$v_{-k}$};
        \path (v2.north) node[anchor=south] {$v_{-k+1}$};
        \path (vk.north) node[anchor=south] {$v_{-1}$};
        \path (r.south east) node[anchor=north] {$v_{0}$};
        \path (v1p.south) node[anchor=north] {$v_{1}$};
        \path (v2p.south) node[anchor=north] {$v_{2}$};
        \path (vkp.south) node[anchor=north] {$v_{k}$};
    \end{pgfinterruptboundingbox}

    \path[edge,inner,north east loop] (r) to (r);
\end{tikzpicture}%
\vspace{-1.6em}}

\bigskip\bigskip

\subfloat[]{\label{f:high-peripherality:no-finitary}\hspace*{3.5ex}\begin{tikzpicture}
    \newcommand{\dist}{10mm}\input{figure/graphstyle}

    \path node[vertex,outer] (v1) {} 
            ++(0:\dist) node[vertex, outer] (v2){}
            ++(0:\dist) node[vertex, outer,opacity=0] (v3){}
            ++(0:1.25*\dist) node[vertex, outer,opacity=0] (vkmo){}
            ++(0:\dist) node[vertex, outer] (vk) {}
            ++(-90:\dist) node[vertex, inner] (r) {}
            ++(-90:\dist) node[vertex, outer] (v1p) {}
            ++(180:\dist) node[vertex, outer] (v2p) {}
            ++(180:\dist) node[vertex, outer,opacity=0] (v3p) {}
            ++(180:1.25*\dist) node[vertex, outer,opacity=0] (vkmop) {}
            ++(180:\dist) node[vertex, outer] (vkp) {}
    ;
    
    \path (r) ++(150:\dist) node[vertex, inner] (rsource) {};
    \path (r) ++(-150:\dist) node[vertex, inner] (rtarget) {};

    \path[edge,outer] (v1) to (v2);
    \path[edge,outer] (v2) to (v3);
    \path (v3) to node[midway,outer] {$\cdots$} (vkmo);
    \path[edge,outer] (vkmo) to (vk);
    \path[edge,outer] (vk) to (r);
    \path[edge,outer] (r) to (v1p);
    \path[edge,outer] (v1p) to (v2p);
    \path[edge,outer] (v2p) to (v3p);
    \path (v3p) to node[midway,outer] {$\cdots$} (vkmop);
    \path[edge,outer] (vkmop) to (vkp);

    \path[edge,inner] (rsource) to (r); 
    \path[edge,inner] (r) to (rtarget);
    \path[edge,inner,west loop] (rsource) to (rsource);
    \path[edge,inner,west loop] (rtarget) to (rtarget);
\end{tikzpicture}%
\vspace{-1.6em}}\hspace*{12mm}\subfloat[]{\label{f:high-peripherality:low-orderability}\hspace*{1ex}\begin{tikzpicture}
    \newcommand{\dist}{10mm}\input{figure/graphstyle}

    \path 
node[vertex, finitary] (v2){}
            ++(0:\dist) node[vertex, finitary,opacity=0] (v3){}
            ++(0:1.25*\dist) node[vertex, finitary,opacity=0] (vkm1){}
            ++(0:\dist) node[vertex, finitary] (vk) {}
            ++(0:\dist) node[vertex, finitary] (vkp1) {}
            ++(0:\dist) node[vertex, finitary, opacity=0] (vkp2){}
            ++(0:1.25*\dist) node[vertex, finitary, opacity=0] (v2k){}
            ++(0:\dist) node[vertex, finitary] (v2kp1){}
    ;

    \path[edge,finitary] (v2) to  (v3);
    \path (v3) to node[midway,finitary] {$\cdots$} (vkm1);
    \path[edge,finitary] (vkm1) to (vk);
    \path[edge,finitary] (vk) to node[pos=.4,above] {$a$} (vkp1);
    \path[edge,finitary] (vkp1) to (vkp2);
    \path (vkp2) to node[midway,finitary] {$\cdots$} (v2k);
    \path[edge,finitary] (v2k) to (v2kp1);

    \path (v3) ++(-90:\dist) 
node[vertex, outer] (u2){}
            ++(0:\dist) node[vertex, outer,opacity=0] (u3){}
            ++(0:1.25*\dist) node[vertex, outer,opacity=0] (ukm1){}
            ++(0:\dist) node[vertex, outer] (uk) {}
            ++(0:\dist) node[vertex, inner] (ukp1) {};

    \path[edge,finitary] (uk) to node[pos=.4,right] {$b$} (vkp1);
    
    \path[edge,outer] (u2) to (u3);
    \path (u3) to node[midway,outer] {$\cdots$} (ukm1);
    \path[edge,outer] (ukm1) to (uk);
    \path[edge,outer] (uk) to node[pos=.4,below] {$c$} (ukp1);

\path[edge,inner,east loop] (ukp1) to node[right] {$d$} (ukp1);

    \draw [decorate,decoration={brace,amplitude=.6em,raise=.75em}]
        (v2) -- (vk) node[midway,yshift=2em]{\small a walk $w_1$ s.t.~$|w_1|=k$};
    \draw [decorate,decoration={brace,amplitude=.6em,raise=.75em}]
        (vkp1) -- (v2kp1) node[midway,yshift=2em]{\small a walk $w_2$ s.t.~$|w_2|=k$};
    \draw [decorate,decoration={brace,amplitude=.6em,mirror,raise=0.75em}]
        (u2) -- (uk) node[midway,yshift=-2em]{\small a walk $w_1'$ s.t.~$|w_1'|=k$};
    
\end{tikzpicture}\vspace{-1.95em}}\caption{Families of (inner-connected cyclic) graphs with high peripherality. Graphs in family~(b) moreover have a high eventual peripherality.}
\label{f:high-peripherality}
\end{figure}

\subsection{Walks with Peripherality Zero}
\label{s:peripherality:zero}
\begin{toappendix}
    \subsection{Walks with Peripherality Zero}
\end{toappendix}

Next, we focus on walks with peripherality~$0$, i.e., of type~$\L^?\I^?$
or~$\I^?\R^?$.  We show that, in inner-connected graphs, any two such walks of equal length are related
by~$\ppreleq*2$, as stated next. This result will be used 
afterwards to show the upper bounds of~\cref{t:pp-rank},~\cref{t:pp-rank:classic,t:pp-rank:incr,t:pp-rank:eve-incr}.

\begin{propositionrep}\label{prop:inner-connected=>peripherality-0-connected}
    Let~$G$ be an inner-connected graph and~$w,w'$ two walks in~$G$.
    If~$w$ and~$w'$ have peripherality~$0$ and~$|w|=|w'|$,
    then~$w \ppreleq*2 w'$.
\end{propositionrep}
\begin{proofsketch}
    Let~$S$ be the set of walks of~$G$ that entirely belong to a single nontrivial scc;
    the scc is not fixed, so that $w \in S$ whenever there is a nontrivial scc
    containing every edge of~$w$. \cref{prop:inner-connected=>peripherality-0-connected} follows directly from the \cref{clm:strong-crawling,clm:caterpillar}, stated next.

    \begin{claim}\label{clm:strong-crawling}
        For every nonempty walk~$w$ of peripherality~$0$, there exists a
        walk~$w'\in S$ such that~$w\ppreleq*{2}w'$.
    \end{claim}
    \begin{claim}
        \label{clm:caterpillar}
        For any two walks~$w,w'\in S$ of the same length, we have~$w \ppreleq*2 w'$.
    \end{claim}

    \cref{clm:strong-crawling} relies on the fact that any walk~$w$ with a peripherality of~0 is the prefix (or the suffix, which we handle symmetrically) of a walk~$w\cdot w''$ that ends in a nontrivial scc~$C$. Given any walk~$w'$ that is entirely in~$C$, starts in~$\tgt(w\cdot w'')$ and satisfies~$|w|=|w'|$, the claim follows from the fact that consecutive factors of length~$|w|$ in~$w\cdot w'' \cdot w'$ are related by~$w\ppreleq{2}w'$.
    Now, for the proof of \cref{clm:caterpillar}, the reasoning depends on the relationship between the scc~$C$ associated with~$w$  and the scc $C'$ associated with~$w'$.
    In the special case where $C'$ is reachable from~$C$ (including $C=C'$), we show \cref{clm:caterpillar} just like \cref{clm:strong-crawling}, using the fact that there is a walk~$w''$ such that~$w\cdot w''\cdot w'$ is a valid walk. This special case corresponds to the case where~$C$ can reach~$C'$ in the 
    condensation~$D$
of the inner subgraph~$I$ of~$G$. 
    By iterating this reasoning, we show \cref{clm:caterpillar} in the general case by following a chain connecting~$C$ and~$C'$ in~$D$, which exists thanks to our assumption that $I$ is weakly connected (because $G$ is inner-connected).
\end{proofsketch}

\begin{proof}
Let~$S$ be the set of walks of~$G$ that entirely belong to a single nontrivial scc;
the scc is not fixed, so that $w \in S$ whenever there is a nontrivial scc
containing every edge of~$w$. The proof proceeds in two steps: every walk of
peripherality~$0$ is related by~$\ppreleq*2$ to some walk in~$S$
(\cref{clm:strong-crawling}), and any two walks of~$S$ of the same length are
related by~$\ppreleq*2$ (\cref{clm:caterpillar}). Both steps rely
on~\cref{clm:crawling}, which we prove first.

\begin{claim}\label{clm:crawling}
    Let~$w_1,w_2$ be two walks of the same length.
    If the source of~$w_2$ is reachable in $G$ from the target of~$w_1$,
    then~$w_1\ppreleq*{2}w_2$.
\end{claim}
\begin{claimproof}[Proof of \cref{clm:crawling}]
    We let~$n$ denote the length of~$w_1$ and~$w_2$.
    The hypotheses imply the existence of a walk~$w'$ such that $w=w_1\cdot w' \cdot w_2$ is a walk of~$G$.
    Let~$k$ be the length of~$w'$.
    For every~$i$ in $\{0, \ldots, n+k\}$ we denote by~$w=p^{(i)}w^{(i)}s^{(i)}$
    the unique factorization of~$w$ such that~$p^{(i)}$ is of length~$i$
    and~$w^{(i)}$ is of length~$n$.
    Note that~$w^{(0)}=w_1$ and~$w^{(n+k)}=w_2$.
    Since a pop-left and a push-right transform $w^{(i)}$ into $w^{(i+1)}$, we have~$w^{(i)} \ppreleq2 w^{(i+1)}$, for every~$i$, hence the statement follows.
\end{claimproof}

\begin{falseclaim}{\ref{clm:strong-crawling}}
    For every nonempty walk~$w$ of peripherality~$0$, there exists a
    walk~$w'\in S$ such that~$w\ppreleq*{2}w'$.
\end{falseclaim}
\begin{claimproof}[Proof of \cref{clm:strong-crawling}]
    A walk of peripherality~$0$ has type $\L^?\I^?$  or $\I^?\R^?$.
    By symmetry we only focus on the case where $w$ is a (nonempty)
$\I^?\R^?$-walk.
    Thus, by definition, $\src(w)$ satisfies \cref{eq:kind:after-scc}. Hence there exists a walk~$w''$ such that~$w''\cdot w$ starts in a nontrivial
    scc~$C$. Let~$w'$ be any walk of length~$|w|$ that ends at $\src(w'')$ and
    that entirely belongs to~$C$, as can be obtained following any cycle of the nontrivial scc~$C$. The claim follows by applying \cref{clm:crawling} to $w'$ and $w$.
\end{claimproof}

\begin{falseclaim}{\ref{clm:caterpillar}}
    For any two walks~$w,w'\in S$ of the same length, we have~$w \ppreleq*2 w'$.
\end{falseclaim}
\begin{claimproof}[Proof of~\cref{clm:caterpillar}]\newcommand{\locrel}[1][]{\mathrel{\sim^{#1}}}
    Let~$D$ be the condensation 
of the inner subgraph of~$G$.
    Let $\locrel$ be the relation on nontrivial sccs of $G$ defined by
    $C \locrel C'$ if there exists a path in~$D$ from $C$ to $C'$ or vice versa. 
    Since $G$ is inner-connected,~$D$ is weakly connected, hence $C \locrel[*] C'$ holds for every sccs~$C,C'$.
    
    The claim is obtained by an induction on~$n$ showing that $C \locrel[\leq n] C'$ implies that we have $w \ppreleq*2 w'$, for every walk $w$ that entirely belongs to~$C$ and $w'$ that entirely belongs to~$C'$ and has the same length as~$w$.
    \begin{itemize}
        \item The base case ($n=0$) implies~$C=C'$ hence~$w$ and~$w'$ both entirely belong to the same scc~$C$. We can then conclude using \cref{clm:crawling}
          because the source of $w'$ and the target of~$w$ are in the same
          scc~$C$.
        \item For the inductive step, assume $C \locrel[\leq n] C'' \locrel C'$,
          for some scc~$C''$. Let $w''$ be a walk of length $|w|$ that entirely belongs 
          to~$C''$.
        The inductive hypothesis shows 
that~$w \ppreleq*2 w''$.
        On the other hand, $C'' \locrel C'$ implies by definition that~$C'$ is reachable from~$C''$ in~$D$, or vice-versa.
        In the former case, the source of~$w'$ is reachable from the target of~$w''$.
        In the latter case, the source of~$w''$ is reachable from the target of~$w'$.
        In both cases, applying \cref{clm:crawling} yields~$w''\ppreleq*2 w'$.
        We conclude by the transitivity of
        $\ppreleq*2$.
\renewcommand\qedsymbol{\textcolor{lipicsGray}{\ensuremath{\vartriangleleft}}}\qedhere
    \end{itemize}
\end{claimproof}

\cref{prop:inner-connected=>peripherality-0-connected} then follows directly from 
    \cref{clm:strong-crawling,clm:caterpillar}.
\end{proof}

\subsection{Upper Bounds}
\label{s:peripherality:upper}
\begin{toappendix}
\subsection{Upper Bounds}
\end{toappendix}

We now show the main results from this section: upper bounds for pp-rank (\cref{p:inner-connected=>3-orderable}), pp-increasing-rank (\cref{p:inner-connected=>6k-increasing-orderable}) and eventual pp-increasing-rank (\cref{p:inner-connected=>eventual-6k-increasing-orderable}).

\begin{propositionrep}\label{p:inner-connected=>3-orderable}
    Every inner-connected graph is 3-orderable.
\end{propositionrep}
\begin{proofsketch}
    First note that, for any graph $G$ (cyclic or acyclic), the graph $\walks{G}$ is $\pprel1$-connected (every walk is $\pprel1^*$-related to the empty walk). Thus, when $G$ is acyclic then $\walks{G}$ is finite and \cref{t:even-odd} applies. We assume from now on that $G$ is cyclic.

    For every nonnegative integer~$k$, we write~$Y_k=\{ w\in \walks{G} \mid \peri(w)=0 \text{ and } |w|=k\}$.
    The following claim can be shown using \cref{prop:inner-connected=>peripherality-0-connected}.

    \begin{claim}\label{cl:inner-connected=>3-orderable-1}
        For every nonnegative integer~$k$,~$(Y_k\cup Y_{k+1}, \pprel1)$ is a connected graph.
    \end{claim}

     Now, we write~$Z_k = \{ w\in \walks{G} \mid \big( |w|-\peri(w)\big) = k \}$, for every~$k\in\N$.
     Notice that~$Z_k\supseteq Y_k$ and~$\biguplus_{i\in\N} Z_{i}=\walks{G}$.
     Note moreover that given a walk~$w$ such that~$\peri(w)>0$
     then either the prefix or the suffix of~$w$ of length~$\big(|w|-1\big)$, which we write~$w'$, has a peripherality of~$\big(\peri(w)-1\big)$.
     In particular,~$w\in Z_k$ implies that~$w'\in Z_k$, and by iterating this reasoning we can reach $Y_k$, so \cref{cl:inner-connected=>3-orderable-2} follows.
     
    \begin{claim}\label{cl:inner-connected=>3-orderable-2}
        For every nonnegative integer~$k$,~$(Z_k\cup Z_{k+1}, \pprel1)$ is a connected graph.
    \end{claim}

    Since~$G$ is cyclic, it contains an infinite walk that is of type~$\I$, e.g., following any cycle in a nontrivial scc. We denote by~$w_i$ the prefix of length~$2i$ of this walk, for~$i\in\N$; note that $w_i \in Y_{2i} \subseteq Z_{2i}$. 
    We write~$U=(\walks{G},\pprel{1})$ and,~$V_i=Z_{2i}\cup Z_{2i+1}$, for~$i\in\N$.
    Applying \cref{p:layer:even-odd} to $U$, $(V_i)_{i\in\N}$ and~$(w_i)_{i\in\N}$ yields a 3-bounded ordering of $\walks{G}$.
\end{proofsketch}
\begin{proof}
    If~$G$ is acyclic then the undirected graph~$U=(\walks{G},\pprel{1})$ is finite.  Moreover, $U$ features a chain from any walk~$w\in\walks{G}$ to the empty walk using the prefixes of~$w$, hence~$U$ is connected.
    A direct application of~\cref{t:even-odd} to~$U$ concludes the proof in this case.
    
    We assume~$G$ to be cyclic in the remainder of the proof.
    For every nonnegative integer~$k$, we define~$Y_k$ as follows.
    \begin{equation}
        Y_k=\{ w\in \walks{G} \mid \peri(w)=0 \text{ and } |w|=k\}
    \end{equation}
    We show the following claim using \cref{prop:inner-connected=>peripherality-0-connected}.

    \begin{falseclaim}{\ref{cl:inner-connected=>3-orderable-1}}
        For every nonnegative integer~$k$,~$(Y_k\cup Y_{k+1}, \pprel1)$ is a connected graph.
    \end{falseclaim}
    \begin{claimproof}[Proof of \cref{cl:inner-connected=>3-orderable-1}]
      Notice that $(Y_k\cup Y_{k+1}, \pprel1)$ is a bipartite graph. We first show the following: 
      \begin{itemize}
          \item[$(\dagger)$\hspace*{-2mm}] for~$w_1$ and~$w_2$ two walks in~$Y_{k+1}$ such that~$w_1\ppreleq{2}w_2$, we have that
        $w_1$ is connected to~$w_2$ in~$(Y_k\cup Y_{k+1},\pprel{1})$.
      \end{itemize}
Since~$k+1\geq 1$, there is a (possibly empty) walk~$w'$ and two edges~$e_1$ and~$e_2$ such that~$w_1$ and~$w_2$ are of one of the following shapes ($e_1$ is the edge that we remove and $e_2$ the one we add to go from $w_1$ to $w_2$):
        \begin{gather}
            w_1=e_1\cdot w' \text{~~or~~} w_1=w' \cdot e_1 \label{eq:t:3-orderable-1:i}\\
            w_2=e_2\cdot w' \text{~~or~~} w_2=w' \cdot e_2
            \label{eq:t:3-orderable-1:ii}
        \end{gather}
        In particular, $w_1\pprel{1} w' \pprel{1} w_2$.
        Note that $|w'|=k$ and since $\peri(w_1)=0$, one can check that $\peri(w')=0$ as well, so that we have~$w'\in Y_k$.  Hence~$w_1$ and~$w_2$ are indeed connected in~$(Y_k\cup Y_{k+1}, \pprel1)$, which concludes the proof of $(\dagger)$.
    
        To finish the proof of \cref{cl:inner-connected=>3-orderable-1}, let~$w$ and~$w'$ now be arbitrary walks in~$Y_k\cup Y_{k+1}$.  
        If~$|w|=k$, then since~$w$ is of peripherality~$0$, for reasons similar to the proof of \cref{clm:strong-crawling} there exists an edge~$e$ such that~$w''=w \cdot e$ or~$w''=e\cdot w$ is a walk of length~$k+1$ and of peripherality~$0$ which satisfies~$w\pprel{1} w''$.
        The same holds for~$w'$. Therefore, to show \cref{cl:inner-connected=>3-orderable-1}, it is enough to prove that for any two walks $w,w'$ of $Y_{k+1}$, these walks are connected in $(Y_k\cup Y_{k+1}, \pprel1)$.
        \cref{prop:inner-connected=>peripherality-0-connected} yields that~$w\ppreleq*{2} w'$, hence there exist walks~$w_1,\ldots,w_n$ such that:
        \begin{equation*}
            w=w_1\ppreleq{2}w_2 \ppreleq{2} \cdots \ppreleq{2} w_n=w'
        \end{equation*}
        Applying $(\dagger)$ to each pair~$(w_i,w_{i+1})$ we obtain that~$w_i$ is connected to~$w_{i+1}$ in~$(Y_k\cup Y_{k+1},\pprel{1})$. 
        It follows that~$w$ is connected to~$w'$ in~${(Y_k\cup Y_{k+1},\pprel{1})}$, concluding the proof of \cref{cl:inner-connected=>3-orderable-1}.
    \end{claimproof}

     Now, we write~$Z_k = \{ w\in \walks{G} \mid \big( |w|-\peri(w)\big) = k \}$, for every~$k\in\N$.
     Notice that~$Z_k\supseteq Y_k$ and~$\biguplus_{i\in\N} Z_{i}=\walks{G}$.
     We then extend \cref{cl:inner-connected=>3-orderable-1} as follows.
     
    \begin{falseclaim}{\ref{cl:inner-connected=>3-orderable-2}}
        For every nonnegative integer~$k$,~$(Z_k\cup Z_{k+1}, \pprel1)$ is a connected graph.
    \end{falseclaim}
    \begin{claimproof}[Proof of \cref{cl:inner-connected=>3-orderable-2}]
        For every walk~$w$ in~$Z_k\cup Z_{k+1}$ with~$\peri(w)>0$, there is an edge~$e$ and a walk~$w'$ such that~$\peri(w')=\peri(w)-1$ and either~$w=e\cdot w'$ or~$w=w'\cdot e$.
        It may be verified that~$w\pprel1 w'$ and~$w'\in Z_k\cup Z_{k+1}$.
        Inductively applying this reasoning shows that there is a walk~$w''\in Y_k\cup Y_{k+1}$ that is connected to~$w$ in~$(Z_k\cup Z_{k+1}, \pprel1)$.
        Now, since $(Y_k \cup Y_{k+1},\pprel1)$ is connected by~\cref{cl:inner-connected=>3-orderable-1}, the proof of \cref{cl:inner-connected=>3-orderable-2} follows.
    \end{claimproof}

    We will now finish the proof of \cref{p:inner-connected=>3-orderable} by applying \cref{p:layer:even-odd} as follows. 
    We write~$U=(\walks{G},\pprel{1})$ and~$V_i=Z_{2i}\cup Z_{2i+1}$, for each~$i\in\N$. 
    Since $(Z_i)_{i\in \N}$ is a partition of $\walks{G}$,  $(V_i)_{i \in \N}$ is also a partition. By \cref {cl:inner-connected=>3-orderable-2}, $(V_i,\pprel{1})$ is connected for every~$i\in\N$ thus condition \ref{p:layer:even-odd}\ref{p:layer:even-odd:pre:subgraph-connected} holds.
    Since~$G$ is cyclic, it contains an infinite walk~$w$ that is of type~$\I$, e.g., following any cycle in a nontrivial scc. 
    For every~$i\in \N$, we write~$w_i$ for the prefix of~$w$ of length~$2i$.
    The sequence~$(w_i)_{i\in\N}$ thus defined satisfies Condition~\ref{p:layer:even-odd}\ref{p:layer:even-odd:pre:direct-seq}: we have~$w_i\in Y_{2i}\subseteq Z_{2i}\subseteq V_i$ and $w_i\left(\pprel{1}\right)^{2} w_{i+1}$, for every~$i\in\N$.
    \cref{p:layer:even-odd} thus yields an ordering of $\walks{G}$ that is 3-bounded.   
\end{proof}

The ordering~$w_1,\ldots,w_n,\ldots$ we build in the proof of \cref{p:inner-connected=>3-orderable} is not increasing, but it is in some sense \emph{quasi-increasing}.
Indeed, the difference in length between a walk~$w_{\ell+i}$ and any walk~$w_{\ell}$ coming before has a constant lower bound:
\begin{equation}
\forall \ell,i\in\N\qquad|w_{\ell+i}|-|w_\ell{}| \geq \underbrace{-\peri(G)-1}_{(\ddagger)}
\end{equation}
In order to be increasing, the ordering would need to satisfy (\theequation)
with~$0$ instead of~$(\ddagger)$, something that is never ensured by \cref{p:inner-connected=>3-orderable} since~$\peri(G)\geq 0$.

We have now shown the upper bound of \cref{t:pp-rank},~\cref{t:pp-rank:classic}, which improves the bound of \cite{amarilli2023enumerating}. We now show the upper bounds
of \cref{t:pp-rank:incr,t:pp-rank:eve-incr}, which focus on 
ordering by increasing length, and which are shown using \cref{p:incr-len-orderability<=layer-orderability}.

\begin{propositionrep}\label{p:inner-connected=>6k-increasing-orderable}
    Every inner-connected graph~$G$ is~$6k$-orderable by increasing length, where~$k=\max(1,\peri(G))$.
\end{propositionrep}
\begin{proofsketch}
    From \cref{p:incr-len-orderability<=layer-orderability}, we only have to prove that~${(\walks[\ell]{G}, \ppreleq{2k})}$ is connected for every~$\ell$.
    This was established in the case where $\peri(G)=0$ by \cref{prop:inner-connected=>peripherality-0-connected}, and it is obvious in the case where~$G$ is acyclic.
    The remaining case of a cyclic graph~$G$ with $\peri(G)>0$ is handled by \cref{cl:inner-connected=>6k-increasing-orderable:cyclic} below. The claim is proved with a simple case analysis on whether~$w$ is of type~$\L\I^?\R$ or~$\L^?\F\R^?$.
\begin{claim}\label{cl:inner-connected=>6k-increasing-orderable:cyclic}
        If~$G$ is cyclic, then for every walk~$w$ in $G$ with positive peripherality, there exists a walk~$w'$ in~$G$ such that~$\peri(w')=0$ and~$w \ppreleq{2k} w'$.
    \end{claim}
    
    Using~\cref{prop:inner-connected=>peripherality-0-connected} concludes the proof.
\end{proofsketch}
\begin{proof}
    Thanks to \cref{p:incr-len-orderability<=layer-orderability}, it is sufficient to show the fact that the undirected graph~${(\walks[\ell]{G}, \ppreleq{2k})}$ is connected for every~$\ell$. If $\peri(G)=0$ then \cref{prop:inner-connected=>peripherality-0-connected} concludes, so let us assume that $\peri(G)>0$.
    If $G$ is acyclic, the size of all walks is bounded by the peripherality thus any two walks of the same size $w,\,w'$ satisfy $w\ppreleq{2k}w'$ and so ${(\walks[\ell]{G}, \ppreleq{2k})}$ is indeed connected for every~$\ell$. We are left with the case of cyclic graphs of positive peripherality.
    We show the following claim.
    
    \begin{falseclaim}{\ref{cl:inner-connected=>6k-increasing-orderable:cyclic}}
        If~$G$ is cyclic, then for every walk~$w$ in $G$ with positive peripherality, there exists a walk~$w'$ in~$G$ such that the following two conditions hold.
    \begin{align}\label{eq:main-claim}
        \peri(w')=0 && w \ppreleq{2k} w'
    \end{align}
    \end{falseclaim}
    \begin{claimproof}[Proof of \cref{cl:inner-connected=>6k-increasing-orderable:cyclic}]
    Since $w$ is of positive peripherality, it is either of type~$\L\I^?\R$ or $\L^?\F\R^?$.
    \begin{itemize}
        \item
        Case~1: $w$ can be $\L\I^?\R$-factorized as~$w=w_\ell\cdot w_i\cdot w_r$.
        We assume without loss of generality that~$\peri(w)=|w_\ell|\leq k$
        because the other case ($|w_r|\leq k$) is symmetric.
        The first vertex~$v$ of~$w_i\cdot w_r$ 
        is necessarily inner,
        because $w_\ell$ is nonempty and of type $\L$ and continues with $w_i \cdot w_r$ which is also nonempty and of type $\I^?\R$.
        Hence, there are 
arbitrarily long walks that start in a nontrivial scc~$C$ and end at~$v$. All edges of such walks are inner, because they lie on a path from $C$ to~$v$ which is inner and hence can reach a nontrivial scc. 
Take such a walk of length $\geq |w_\ell|$, and let $w_{\ell'}$ be its suffix of length $|w_\ell|$.
The walk $w' \coloneq w_{\ell'} \cdot w_i \cdot w_r$ is of type $\I\R$ so it has peripherality~$0$ and we clearly have $|w|= |w'|$. Moreover, the walks differ by changing a prefix of length $\leq k$ so indeed  $w \ppreleq{2k} w'$ and we have established
        \eqref{eq:main-claim}.
\item 
        Case~2: $w$ can be~$\L^?\F\R^?$-factorized as~$w=w_\ell\cdot w_f\cdot w_r$. If both~$w_\ell$ and~$w_r$ are empty,  then~$|w|=\peri(w)\leq k$, hence any walk~$w'$ such that~$|w|=|w'|$ satisfies~$w \ppreleq{2k} w'$. Since $G$ is cyclic we can choose $w'$ of peripherality~$0$,  hence \eqref{eq:main-claim} holds.
        
         Otherwise, either~$w_r$ or~$w_\ell$ is nonempty. We assume without loss of generality that~$|w_r|\geq|w_\ell|$, hence that~$w_r$ is nonempty; the other case is symmetric.
        It follows that~$\peri(w)=\big(|w_\ell|+|w_f|\big)\leq k$.
        Since~$w_r$ is of type~$\R$, its origin is reachable from a nontrivial scc, so there exist two walks~$m_1$ of type $\I^?$, $m_2$ of type~$\R^?$ such that~$w'=m_1\cdot m_2 \cdot w_r$ is a walk 
        and~$|m_1\cdot m_2| = |w_\ell\cdot w_f|$.
        A simple verification then shows that~$w'$ satisfies~\eqref{eq:main-claim}.         
        \renewcommand\qedsymbol{\textcolor{lipicsGray}{\ensuremath{\vartriangleleft}}}\qedhere
    \end{itemize}    \end{claimproof}  
    
    Thus, any walk of positive peripherality is $\ppreleq{2k}$-related to one of the same length and of peripherality $0$. By~\cref{prop:inner-connected=>peripherality-0-connected}, the set of walk of peripherality~$0$ and length $\ell$ is connected for $\ppreleq{2}$. 
    This implies that $(\walks[\ell]{G}, \ppreleq{2k})$ is connected for
    every~$\ell$ and concludes the proof.
\end{proof}

\begin{proposition}\label{p:inner-connected=>eventual-6k-increasing-orderable}
    Every inner-connected graph~$G$ is eventually~$6k'$-orderable by increasing length, where~$k'=\max(1,\evperi(G))$.
\end{proposition}

The proof is the same as that of \cref{p:inner-connected=>6k-increasing-orderable}, except that we consider only the walks that are 
longer 
than~$|V|$.
  This ensures by the pigeonhole principle that only Case~1 in the proof of (the analogue of) \cref{cl:inner-connected=>6k-increasing-orderable:cyclic} can occur, and then the walk $w$ under consideration encounters a nontrivial scc so its peripherality is at most the eventual peripherality of~$G$, so we can indeed use~$k'$ instead of~$k$ in the statement.

\section{Criticality and Lower Bounds}
\label{s:criticality}
\begin{toappendix}
    \label{a:criticality}
\end{toappendix}

Having shown the upper bounds on \cref{t:pp-rank}, we now turn to the lower bounds.
These hinge on the notions of \emph{criticality} and \emph{critical walks}, on the fact that critical walks are $\ppreleq{2k}$-disconnected from the other walks when~$k$ is the criticality of the graph.
We give the relevant definitions and show this result in \cref{s:criticality:def}. We then show the lower bounds of \cref{t:pp-rank} in \cref{sec:lowercrit}, and discuss their tightness in \cref{sec:tightness}.

\subsection{Criticality and Critical Walks} 
\label{s:criticality:def}
\begin{toappendix}
    \subsection{Criticality and Critical Walks} 
    \label{a:criticality:def}     
\end{toappendix}

\begin{definition}\label{d:criticality}
  We define the \defemph{criticality} $\crit(w)$ of a walk $w$ as follows.
  If~$w$ is of type~$\L^?\I^?\R^?$, then~$\crit(w)=\peri(w)-1$.
    Otherwise (i.e.,~$w$ is of type~$\L^?\F\R^?$), its criticality is defined
    as the maximal number~$k$
    such that~$w$ can be factorized as~$w=w_1\cdot e\cdot w_2$
    where~$w_1$ is an $\L^? \F^?$-walk, $w_2$ is an $\F^? \R^?$-walk, $e$ is a single edge of type $F$ and $\min\big(|w_1|,|w_2|\big)\geq k$.

    We also define the \defemph{criticality} of a graph~$G$ as the maximal criticality among walks in~$G$, and the \defemph{eventual criticality} of~$G$ as the maximal criticality among walks that encounter a nontrivial scc (or $-1$ if there are no such walks, i.e., precisely when the graph is acyclic).
\end{definition}

The next property can be shown by an easy induction over the peripherality of walks.

\begin{propertyrep}\label{pp:peripherality-vs-criticality}
  Every walk or graph with peripherality~$p$ and criticality~$c$ satisfies~$c+1 \leq p \leq 2c+2$.
  Every graph with eventual peripherality~$p'$ and eventual criticality~$c'$ satisfies~$c'+1 = p'$.
\end{propertyrep}
\begin{proof}
  We first show the claim about walks, and turn to the claims about graphs at the end of the proof. Let $w$ be a walk. There are two cases.

    First, if~$w$ is a~$\L^?\I^?\R^?$-walk, then by definition we have~$c+1=p\leq 2c+2$.
    
    Second, if~$w$ is a~$\L^?\F\R^?$-walk, let us show that we have~$c+1 \leq p \leq 2c+2$ by induction over the criticality of~$w$ (which cannot be $<0$ according to \cref{d:criticality} because $w$ is of type $\L^?\F\R^?$).
    If~$c=0$, then~$w$ is of the form~$e\cdot w'$ (or symmetrically of the form~$w'\cdot e$) where either: $w'$ contains only outer elements, in which case~$p=1$; or~$w'$ consists of exactly one~$\F$-edge, in which case~$p=2$. In particular, we indeed have~$c+1 \leq p \leq 2c+2$.
    If~$c\geq 1$, then we have~$|w|>2$ and we factorize~$w$ as~$w=e\cdot w' \cdot e'$ for some walk~$w'$ and edges~$e,e'$.
    Since~$c\geq 1$, $w'$ contains an \F-edge and by definition of criticality, we have~$\crit(w')=c-1$.
    We apply the induction hypothesis which yields $c\leq \peri(w') \leq 2c$.
    On the other hand, it can be observed that either: $\peri(w)=\peri(w')+2$ (in the case where~$e$ and~$e'$ both are \F-edges); or~$\peri(w)=\peri(w')+1$ (in all other cases).
    Hence~
    \begin{equation*}
        c+1\leq(\peri(w')+1)\leq p=\peri(w) \leq (\peri(w')+2) \leq 2c+2
    \end{equation*} which concludes the second case and establishes the result for walks.
    
    The statement about graphs for peripherality and criticality follows immediately: if $c+1 \leq p \leq 2c+2$ for every walk, then the same is true for the maxima across all walks. As for eventual peripherality and eventual criticality, the statement follows because the only relevant case is then that of~$\L^?\I^?\R^?$, because sufficiently long walks do not feature any finitary edge.
\end{proof}

\begin{definition}\label{d:critical}
    A walk is \defemph{critical} if it is maximal for the preorder
    defined as $w_1 \preceq w_2$ if:
    \begin{enumerate}[1.]
        \item 
either each of~$w_1$ and~$w_2$ encounter some
          nontrivial scc (not necessarily the same), or neither of them encounters any nontrivial scc; and  
        \item 
        \label{d:critical:cond}
        either $\crit(w_1)<\crit(w_2)$, or both $\crit(w_1)=\crit(w_2)$ and
        $\big(|w_1|_\L+|w_1|_\F+|w_1|_\R\big)\leq\big(|w_2|_\L+|w_2|_\F+|w_2|_\R\big)$.
\end{enumerate}
\end{definition}

The following lemma shows that criticality is preserved when following~$\ppreleq{2k}$:
from a critical walk $w$ of criticality~$k$, one can only reach walks that are equivalent to $w$ wrt. $\preceq$; which intuitively means that these walks are disconnected from the other walks.
\begin{lemmarep}\label{l:crit-walk-connected}
  Let~$w$ be a critical walk with a positive criticality~$k$. Then every walk~$w'$
    such that~$w\ppreleq{2k}w'$ is equivalent to~$w$ w.r.t.~$\preceq$, hence critical.
\end{lemmarep}
\begin{proofsketch}
Since~$\crit(w) > 0$,
the type of~$w$ is either~$\L\I^?\R$ or~$\L^?\F\R^?$.
    We focus in this sketch on the case where~$w$ is of type $\L\I^?\R$, the other case is similar. We $\L\I^?\R$-factorize~$w$ as~$w=w_\ell \cdot w_i \cdot w_r$ with~$|w_\ell|> k$ and~$|w_r|> k$.
    Since~$w\ppreleq{2k} w'$, we have the following factorization of~$w= \underbrace{p_\ell\cdot s_\ell}_{\L} {}\cdot{} \underbrace{w_i}_{\I^?} {} \cdot {} \underbrace{p_r\cdot s_r}_{\R}$ and of~$w'= \underbrace{p_\ell'\cdot s_\ell}_{\L} {}\cdot{} \underbrace{w_i}_{\I^?} {} \cdot {} \underbrace{p_r\cdot s_r'}_{\R}$, for some walks~$p_\ell,p_\ell',s_r,s_r'$ and some \textbf{nonempty} walks~$s_\ell$ and~$p_r$. (Note that~$w_\ell=p_\ell\cdot s_\ell$ and~$w_r=p_r\cdot s_r$.)
    Since~$w$ is critical, it is maximal w.r.t.~$\preceq$, hence greater, equivalent or incomparable to~$p_\ell \cdot s_\ell {}\cdot{} w_i \cdot {} p_r\cdot s_r'$ and to~$p_\ell' \cdot s_\ell {}\cdot{} w_i {} \cdot {} p_r\cdot s_r$.  We are then able to conclude from this that~$|s_r|= |s_r'|$ and~$|p_\ell| = |p_\ell'|$, which finally yields that~$w'$ is equivalent to~$w$ w.r.t~$\preceq$.
\end{proofsketch}

\begin{proof}
Let $w$ be a walk
with~$\crit(w)> 0$.
    This implies that the type of~$w$ is either~$\L\I^?\R$ or~$\L^?\F\R^?$. We distinguish both cases below.

    \medskip
    
    \noindent\textbf{Case where~$w$ is of type $\L\I^?\R$}.
    We $\L\I^?\R$-factorize~$w$ as~$w=w_\ell \cdot w_i \cdot w_r$.
Recall that~$w$ and~$w'$ have the same length.
    Hence, $w\ppreleq{2k} w'$ implies there is~$n\leq k$ such that $n$-pop and~$n$ push operations are sufficient to transform~$w$ into~$w'$.
    Since~$|w_\ell|$ and~$|w_r|$ are both greater than or equal to~$k+1$, 
    it follows that $w_\ell$, $w_r$ and~$w'$ can be factorized as:
    \begin{align*}
        w'= \underbrace{p_\ell'\cdot s_\ell}_{\L} {}\cdot{} \underbrace{w_i}_{\I^?} {} \cdot {} \underbrace{p_r\cdot s_r'}_{\R} && w_\ell = p_\ell\cdot s_\ell && w_r = p_r \cdot s_r
    \end{align*}
    where~$s_\ell$ and~$p_r$ are nonempty.
    As indicated above, since~$s_\ell$ contains only edges of kind~$\L$ and is nonempty,
    it is of type~$\L$,
    hence $p_\ell'\cdot s_\ell$ is also of type~$\L$ (\cref{pp:walk-type}\ref{pp:walk-type:factor}).
    Symmetrically, $p_r\cdot s_r'$ is of type~$\R$
    
    Since~$w$ is critical, it must be greater, equivalent or incomparable to the two walks given below w.r.t.~$\preceq$. Note that their lengths might differ from the length of~$w$.
    \begin{align*}
        w''_1 = \underbrace{w_\ell}_{\L} {}\cdot{} \underbrace{w_i}_{\I^?} {} \cdot {} \underbrace{p_r\cdot s_r'}_{\R} && w''_2 = \underbrace{p_\ell' \cdot s_\ell}_{\L} {}\cdot{} \underbrace{w_i}_{\I^?} {} \cdot {} \underbrace{w_r}_{\R}
    \end{align*}

Both $w''_1$ and $w''_2$ encounter a nontrivial scc if and only if~$w$ does, because vertices of nontrivial sccs can only occur in $w_i$ (or, in the case~$w_i=\varepsilon$, the vertex~$\tgt(w_\ell)=\src(w_r)$ might be part of a nontrivial scc). Indeed the walks $w_\ell$ and $p_\ell' \cdot s_\ell$, respectively $p_r \cdot s_r'$ and $w_r$, are of type $\L$, respectively of type $\R$, so it is not possible that a vertex of a nontrivial scc occurs strictly inside them, hence there cannot be vertices from nontrivial sccs in the parts of the walks that have changed.
    Hence, $w$ is comparable to $w''_1$ and $w''_2$, and it is greater or equivalent to them w.r.t.~$\preceq$. If~$|p_r\cdot s_r'| > |w_r|$, then $\crit(w''_1) \geq \crit(w)$ and $\big(|w''_1|_\L+|w''_1|_\F+|w''_1|_\R\big)>\big(|w|_\L+|w|_\F+|w|_\R\big)$, hence by \cref{d:critical:cond} of  \cref{d:critical}, we would have $w_1''\succ w$, a contradiction.
    Hence, we have $|p_r\cdot s_r'| \leq |w_r|$, and the same reasoning applied to $w''_2$ yields that~$|p_\ell'\cdot s_\ell| \leq |w_\ell|$.
    Moreover, recall that~$|w'| = |w|$, hence we have that
$|p_r\cdot s_r'| + |p_\ell'\cdot s_\ell| = |w_r|+|w_\ell|$
which implies that $|p_r\cdot s_r'|= |w_r|$ and~$|p_\ell'\cdot s_\ell| = |w_\ell|$.
    In other words, we have 
    \begin{equation*}
        \crit(w)=\crit(w'),\qquad |w|_\L=|w'|_\L ,\qquad  |w|_\R=|w'|_\R, \quad\quad \text{and}\quad\quad 0=|w|_\F=|w'|_\F
    \end{equation*}hence $w'$ is equivalent to~$w$, hence critical.

    \medskip

    \noindent\textbf{Case where~$w$ is of type $\L^?\F\R^?$}.
    There exist walks~$w_1,w_2,w_3,w_4,w_1',w_4'$ and an \F-edge~$e$ such that:
    \begin{align}
        \label{c:supercritical:eq1}
        &w=\overbrace{\strut w_1 \cdot w_2}^{\L^?\F^?}{}\cdot{} e{} \cdot{} \overbrace{\strut w_3 \cdot w_4}^{\F^?\R^?} && w'=\overbrace{\strut w_1' \cdot w_2}^{\L^?\F^?} {}\cdot e\cdot {}\overbrace{\strut w_3 \cdot w_4'}^{\F^?\R^?}
        \\ 
        \label{c:supercritical:eq2}
        &|w_1|+|w_4|=|w_1'|+|w_4'|\leq k\\
        \label{c:supercritical:eq3}
        &|w_1\cdot w_2| \geq k, \quad |w_3\cdot w_4| \geq k 
    \end{align}
    \cref{c:supercritical:eq2} comes from the definition of~$\ppreleq{2k}$ and
    \cref{c:supercritical:eq3} from the definition of criticality.
    We know that $w$ and $w'$ cannot encounter a nontrivial scc by definition of the kinds $\L$, $\R$, and $\F$.
    Now, if~$w_1'$ were strictly longer than~$w_1$ then we would have~$(w_1' \cdot w_2\cdot e\cdot w_3\cdot w_4)\succ w$, a contradiction.
Similarly~$w_4'$ cannot be strictly longer than~$w_4$.
    Therefore, \eqref{c:supercritical:eq2} yields that~$|w_1|=|w_1'|$
    and~$|w_4|=|w_4'|$.
    It follows that the criticality of~$w'$ is at least~$k$, hence exactly~$k$ as~$k$ is the maximal criticality.
    Since moreover~$w$ and~$w'$ contain no edge of kind~$\I$, we have~$(|w|_\L+|w|_\F+|w|_\R)=|w|=|w'|=(|w'|_\L+|w'|_\F+|w'|_\R)$,
    hence~$w'$ is equivalent to~$w$ w.r.t.~$\preceq$, hence critical.
\end{proof}

\subsection{Lower Bounds}
\label{sec:lowercrit}
\begin{toappendix}
    \subsection{Lower Bounds}
\end{toappendix}

We now show the main result from \cref{s:criticality}: criticality and eventual criticality 
yield lower bounds for the pp-increasing-rank and eventual pp-increasing-rank, respectively.

\begin{propositionrep}\label{p:criticality-is-lowerbound}
    Let~$G$ be a cyclic graph.
    \begin{enumerate}[a.]
        \item \label{p:criticality-is-lowerbound:incr}
            If the criticality~$c$ of~$G$ is positive, then~$G$ is not~$(2c+1)$-orderable by increasing length.
        \item \label{p:criticality-is-lowerbound:eve-incr}
            If the eventual criticality~$c'$ of~$G$ is positive, then~$G$ is not eventually~$(2c'+1)$-orderable by increasing length.
    \end{enumerate}
\end{propositionrep}

The proof of \cref{p:criticality-is-lowerbound} hinges on \cref{l:incr-len-orderability=>layer-orderability}, stated below, which serves as an approximate converse to \cref{p:incr-len-orderability<=layer-orderability}.  Its proof is straightforward.

\begin{toappendix}

    The proof of \cref{p:criticality-is-lowerbound} hinges on kind of converse to \cref{p:incr-len-orderability<=layer-orderability}, stated next.
\end{toappendix}

\begin{lemmarep}\label{l:incr-len-orderability=>layer-orderability}
    Let~$G$ be a graph and~$k$ a positive integer.
    If~$G$ is $k$-orderable (resp.~eventually $k$-orderable) by increasing length, then the undirected graph $(\walks[\ell]{G}, \ppreleq{k})$ is connected for every~$\ell\in\N$ (resp.\@ for almost every~$\ell\in\N$).
\end{lemmarep}
\begin{proof}
    For orderability by increasing length, let us consider
    a $k$-bounded increasing ordering of~$\walks{G}$. For each $\ell$, we can consider the contiguous subsequence of the walks of length~$\ell$ and it forms a Hamiltonian path of $(\walks[\ell]{G}, \ppreleq{k})$, which thus is connected.
    The proof is similar for eventual orderability by increasing length.
\end{proof}

We can now sketch the proof of \cref{p:criticality-is-lowerbound}:

\begin{inlineproof}[Proof sketch for \cref{p:criticality-is-lowerbound}]
  We rely on the next two claims, shown using
  \cref{l:crit-walk-connected}:
    \begin{claim}\label{cl:criticality-is-lowerbound-layer:finite}
        If~$c>0$, there exists an integer~$\ell$ such that~$(\walks[\ell]{G},\ppreleq{2c+1})$ is disconnected.
    \end{claim}
    \begin{claim}\label{cl:criticality-is-lowerbound-layer:eventual}If~$c'>0$, there exist infinitely many integers~$\ell$ such that~$(\walks[\ell]{G},\ppreleq{2c'+1})$ is disconnected.\end{claim}
    
    \noindent\cref{p:criticality-is-lowerbound:incr} follows from \cref{cl:criticality-is-lowerbound-layer:finite,l:incr-len-orderability=>layer-orderability}, and \cref{p:criticality-is-lowerbound:eve-incr} from \cref{cl:criticality-is-lowerbound-layer:eventual,l:incr-len-orderability=>layer-orderability}.
    The careful reader will have noticed that \cref{cl:criticality-is-lowerbound-layer:finite}
    is stated with~$\ppreleq{2c+1}$ while application of \cref{l:crit-walk-connected} yields the statement for~$\ppreleq{2c}$ only.
    The apparent discrepancy is resolved by the fact
    that~${\ppreleq{2c+1}}={\ppreleq{2c}}$ (see \cref{r:ppreleq-expl}).
\end{inlineproof}

\begin{toappendix}
\begin{proof}[Proof of \cref{p:criticality-is-lowerbound}]
    First, we use \cref{l:crit-walk-connected} to show that critical walks
    are disconnected from other walks in~$(\walks[\ell]{G},\ppreleq{k})$ for relevant~$\ell$ and~$k$.

    \begin{falseclaim}{\ref{cl:criticality-is-lowerbound-layer:finite}}
        If~$c>0$, there exists an integer~$\ell$ such that~$(\walks[\ell]{G},\ppreleq{2c+1})$ is disconnected.
    \end{falseclaim}
    \begin{claimproof}[Proof of \cref{cl:criticality-is-lowerbound-layer:finite}]
        There exists a critical walk $w$ in~$G$ with criticality~$c$: as $c>0$, $w$ is non-empty.
        We write~$\ell=|w|$ with $\ell>0$.
      Since~$G$ is cyclic, it contains an $\I$-walk~$w'$ such that~$|w'|=\ell$; in particular the peripherality of~$w'$ is equal to~$0$, hence its criticality is equal to $-1$
by \cref{pp:peripherality-vs-criticality}, hence~$w'$ is not equivalent to~$w$
      w.r.t.~$\preceq$ because they do not have the same criticality (indeed $w$ has criticality $c>0$).
\cref{l:crit-walk-connected} then yields that~$w$ and~$w'$
        are not related by~$\ppreleq*{2c}$, hence not by~$\ppreleq*{2c+1}$ either since these two relations are equal.
    \end{claimproof}

    \begin{falseclaim}{\ref{cl:criticality-is-lowerbound-layer:eventual}}
        If~$c'>0$, there exist infinitely many integers~$\ell$ such that~$(\walks[\ell]{G},\ppreleq{2c'+1})$ is disconnected.\end{falseclaim}
    \begin{claimproof}[Proof of \cref{cl:criticality-is-lowerbound-layer:eventual}]
        There exists a critical walk~$w$ that has criticality~$c'$
        and encounters a nontrivial scc.
Split $w = w_1 \cdot w_2$ such that the vertex~$x=\tgt(w_1)=\src(w_1)$ belongs to nontrivial scc. 
        Let $w'$ be an arbitrary nonempty walk both starting and ending in~$x$: all edges of~$w'$ are of kind~$\I$. It follows that $w''_m\coloneq w_1 \cdot (w')^m \cdot w_2$ is a walk for every $m \geq 0$. 
        Further, $w$ and $w''_m$ both visit some nontrivial scc (i.e., that of~$x$), neither contains an $\F$-edge so their criticality is their peripherality minus 1 and it is the same because that number only depends on the edges of kind $\L$ and $\R$ of $w$ and $w''_m$ which are the same edges. 
It follows that all walks of the form $w''_m$ are equivalent to~$w$ w.r.t.~$\preceq$ for each $m\geq 0$, and
        since $w$ is critical, they are all critical.
The remainder of the proof is similar to the one of \cref{cl:criticality-is-lowerbound-layer:finite}. 
    \end{claimproof}

    We conclude the proof of \cref{p:criticality-is-lowerbound} by noticing that \cref{p:criticality-is-lowerbound:incr} follows from \cref{cl:criticality-is-lowerbound-layer:finite,l:incr-len-orderability=>layer-orderability} and \cref{p:criticality-is-lowerbound:eve-incr} from \cref{cl:criticality-is-lowerbound-layer:eventual,l:incr-len-orderability=>layer-orderability}.
\end{proof}
\end{toappendix}

Finally, combining \cref{p:criticality-is-lowerbound} with \cref{pp:peripherality-vs-criticality} yields the lower bounds of \cref{t:pp-rank:incr,t:pp-rank:eve-incr} of~\cref{t:pp-rank}:
\begin{enumerate}
    \item  
If~$G$ is a (cyclic or acyclic) graph with at least two edges, then~$(\walks[1]{G},\ppreleq{1})$ is disconnected, hence~$G$ is not 1-orderable by increasing length from~\cref{l:incr-len-orderability=>layer-orderability}. So we immediately have Item~\ref{t:pp-rank}\ref{t:pp-rank:incr}'s lower bound for the pp-increasing-rank of acyclic graphs.
\item Let us show the lower bound of Item~\ref{t:pp-rank}\ref{t:pp-rank:incr} for the pp-increasing-rank of cyclic graphs. 
We distinguish two cases depending on whether $c>0$ or not.

    In the case where $c>0$,
    \cref{p:criticality-is-lowerbound:incr} of \cref{p:criticality-is-lowerbound}
yields that the pp-increasing-rank~$r$ of~$G$ satisfies~$r>2c+1$,
    hence in fact~$r\geq 2(c+1)$ since~$r$ is an integer.
From~\cref{pp:peripherality-vs-criticality}, we have~$2(c+1)\geq p$, hence finally~$r\geq p$ as promised.

    Otherwise, if $c\leq 0$, we have
$p\leq 2$ from \cref{pp:peripherality-vs-criticality}.
    Since~$G$ has at least two edges, we may use the same reasoning as item 1,
    yielding that its pp-increasing-rank is at least~$2$ again, so that $r \geq
    2 \geq p$.
    
    \item We similarly show that the eventual peripherality
      multiplied by two is a lower bound for the eventual pp-increasing-rank (Item~\ref{t:pp-rank}\ref{t:pp-rank:eve-incr}). The theorem statement excludes the case where the only cycle of~$G$ is a self-loop on an isolated vertex, so for every positive length $G$ contains two different walks of that length (using a cycle which is not a self-loop, or two different self-loops, or a self-loop together with one of its incident edges).
We can then use \cref{p:criticality-is-lowerbound:eve-incr} of \cref{p:criticality-is-lowerbound}
instead of \cref{p:criticality-is-lowerbound:incr},
    and equality~$c'+1=p'$ instead of inequality~$2(c+1)\geq p$ (\cref{pp:peripherality-vs-criticality}). The case of $c'\leq 0$ gives $p'\leq 1$ and the claimed lower bound of $2k' = 2$ is achieved like in the case of item~1 but using two walks of the same length for arbitrarily large lengths, whose existence was justified earlier.
\end{enumerate}

\subsection{Lower Bound Tightness}
\label{sec:tightness}
\begin{toappendix}
    \subsection{Lower Bound Tightness}
\end{toappendix}

We conclude \cref{s:criticality} with examples of graph families that meet or
are close to the lower bounds of \cref{t:pp-rank}. Three of the four lower
bounds are met, and the fourth is off by one:
\begin{itemize}
  \item \cref{p:tight}\ref{p:tight:1-orderable} shows that the lower bound in Item~\ref{t:pp-rank}\ref{t:pp-rank:classic} is tight.
  \item \cref{p:tight}\ref{p:tight:2-incr-orderable-dag} shows that the
    lower bound in Item~\ref{t:pp-rank}\ref{t:pp-rank:incr} is tight for acyclic
    graphs.
    \item \cref{p:tight}\ref{p:tight:2p'-eve-incr-orderable} shows that the lower bound in Item~\ref{t:pp-rank}\ref{t:pp-rank:eve-incr} is tight.
  \end{itemize}
  The remaining case is 
the lower bound in Item~\ref{t:pp-rank}\ref{t:pp-rank:incr} for cyclic
graphs:
  \cref{p:tight}\ref{p:tight:p+1-incr-orderable} shows that 
    this bound is almost tight but leaves a gap of one. We leave this gap open, as we do not know whether there is a family of cyclic graphs in which each graph has a pp-increasing-rank 
equal to its peripherality.

\begin{propositionrep}\label{p:tight}
    \begin{enumerate}[a.]
        \item\label{p:tight:1-orderable} There are arbitrarily large cyclic graphs that are 1-orderable.
        \item\label{p:tight:2-incr-orderable-dag} There are arbitrarily large acyclic graphs that are 2-orderable by increasing length.
        \item\label{p:tight:p+1-incr-orderable} There are graphs with arbitrary large peripherality~$p$ that are $(p+1)$-orderable by increasing length.
        \item\label{p:tight:2p'-eve-incr-orderable} There are graphs with arbitrary large eventual peripherality~$p'$ that are eventually $2p'$-orderable by increasing length.
    \end{enumerate}
\end{propositionrep}
\begin{proofsketch}It may be shown that, for every~$k\in\N$:
    \begin{enumerate}[a.]
        \item the cycle of length~$k$ is 1-orderable;
        \item the simple directed path of length~$k$ is 2-orderable by increasing length;
        \item the graph given in \cref{f:high-peripherality:low-orderability} has criticality~$k$, peripherality~$p=2k+1$ and is~$(p+1)$-orderable by increasing length, but not~$p$-orderable by increasing length;
        \item the graph given in \cref{f:high-peripherality:eventual} has eventual peripherality~$k$ and is eventually~$2k$-orderable by increasing length.\qedhere
    \end{enumerate}    
\end{proofsketch}

\begin{toappendix}
    \begin{proof}[Proof of \cref{p:tight}\ref{p:tight:1-orderable}]
        Consider the simple cycle~$C_k$ with~$k$ vertices~$v_0,\ldots,v_{k-1}$, and by abuse of notation we write~$v_j$ for any~$j\in\mathbb{Z}$, meaning~$v_j=v_{j\mod {k}}$. Hence the~$k$ edges of~$C_k$ are~$v_i\rightarrow v_{i+1}$ where~$i$ ranges in~$[k]$.
        For every~$i\in[k]$ and~$j>0$ we denote by~$\langle v_i{:} j\rangle$ the walk in~$C_k$ that starts at~$v_i$ and whose length is~$j$. For every~$i\in\N$ we define the 1-bounded ordering~$O_i$ of~$\walks[2i+1]{C_k}\cup\walks[2i+2]{C_k}$ as follows.
        \begin{equation*}
            \langle v_{k-i}{:} 2i+1\rangle, \langle v_{k-i}{:} 2i+2\rangle, \langle v_{k-i+1}{:} 2i+1\rangle, \langle v_{k-i+1}{:} 2i+2\rangle,\ldots, \langle v_{k-i-1}{:} 2i+1\rangle, \langle v_{k-i-1}{:} 2i+2\rangle
        \end{equation*} 
        Notice that~$O_i$ starts with~$\langle v_{k-i}{:} 2i+1\rangle$
        and ends with~$\langle v_{k-i-1}{:} 2i+2\rangle$.
        Hence, we can start with the empty walk and then connect the layers of the form $\walks[2i+1]{C_k}\cup\walks[2i+2]{C_k}$ by going from
        $\langle v_{k-i-1}{:} 2i+2\rangle$ to $\langle v_{k-i-1}{:} 2i+3\rangle$ at distance $1$. Thus, for every integer~$k\geq 1$, the pp-rank of~$C_k$ is~$1$.
    \end{proof}
    
    \begin{proof}[Proof of \cref{p:tight}\ref{p:tight:2-incr-orderable-dag}]
        Consider the graph~$L_k$, consisting of a single (directed) path of length~$k$. It contains~$(k+1)$ vertices~$v_0,\ldots,v_k$ and~$k$ edges~$v_i\rightarrow v_{i+1}$, for~$i\in[k]$.
Note that~$L_{k}$ has a criticality of~$\lfloor\frac{k-1}{2}\rfloor$ and a peripherality of~$k$.
For every~$i<j$ we denote by~$\langle i{:} j \rangle$ the unique walk in~$L_k$ starting in~$v_i$ and ending in~$v_j$, hence~$\walks{L_k} = \{\varepsilon\}\cup\{\langle i{:}j\rangle \mid i,j\in[k+1], i<j \}$.
For every~$k$, $L_k$ is 2-orderable by increasing length thanks to the following ordering.
        For odd~$\ell$, the set~$\walks[\ell]{L_k}$ is ordered~$\langle 0{:} \ell \rangle,~\langle 1{:} \ell+1 \rangle,\ldots,\langle k-\ell{:} k  \rangle$. For even~$\ell$, we use the inverse order. It yields:
        \begin{multline*}
            \varepsilon,\quad \overbrace{\langle 0{:} 1 \rangle,\langle 1{:} 2 \rangle,\cdots, \langle k-1 {:} k \rangle}^{\walks[1]{L_k}}
            ,\quad
            \overbrace{\langle k-2{:} k \rangle,\langle k-3{:} k-1 \rangle,\cdots, \langle 0 {:} 2 \rangle}^{\walks[2]{L_k}}
            ,\\[5pt]
            \underbrace{\langle 0{:} 3 \rangle,\langle 1{:} 4 \rangle,\cdots, \langle k-3 {:} k \rangle}_{\walks[3]{L_k}}, \quad\cdots,\quad\underbrace{ \langle 0 {:} k \rangle}_{\walks[k]{L_k}} 
        \end{multline*}
        We have now shown the desired claim. We also remark that, in fact, the pp-increasing rank of~$L_k$ is exactly 2 for every~$k\geq 2$, because
        for~$k\geq2$, the graph~$L_k$ has two different walks of the same length hence it is not 1-orderable by increasing length.
    \end{proof}
    
    \begin{proof}[Proof of \cref{p:tight}\ref{p:tight:p+1-incr-orderable}]
        We consider the graph family shown in \cref{f:high-peripherality:low-orderability} and denote by~$G$ the graph of this family with a peripherality of~$2k+1$ and a criticality of~$k$.
Let us show that~$G$ is~$(2k+2)$-orderable by increasing length.
Note that any two walks~$w,w'$ such that~$|w|=|w'|\geq 2k+2$ are of type~$\L^?\I$ and satisfy $w\ppreleq{2k+2} w'$.
Hence, it is quite straightforward to find an increasing ordering~$O_{\geq 2k+2}$ of~$\walks[\geq 2k+2]{G}$ that is~$(2k+2)$-bounded, and that moreover can start by an arbitrary walk of length~$2k+2$.
Since all walks in $\walks[\leq k]{G}$ are related by~$\pprel{2k+2}$,
        it is trivial to build an increasing $(2k+2)$-bounded ordering~$O_{\leq k}$ of $\walks[\leq k]{G}$ that ends with an arbitrary walk of length~$k$.

For every length~$\ell$ such that~$(k+1)\leq\ell\leq(2k+1)$, let us build an ordering~$\overrightarrow{O}_\ell$ of~$\walks[\ell]{G}$.
        We write:
        \begin{align*}
            w = w_1\cdot a\cdot w_2 && w'= w_1' \cdot b \cdot w_2 && w''= w_1''
            \cdot c \cdot d^{\ell}
        \end{align*}
        It may be verified that each walk~$x$ in~$\walks[\ell]{G}$
        is the factor of~$w$, of~$w$ or of~$w''$, and of exactly one of those~since~$\ell\geq (k+1)$.
        We denote these factors as follows. For every~$i\in[2k+2-\ell]$, we denote by~$w_i^{(\ell)}$ the unique walk such that there are~$p,s$ satisfying~$w= p \cdot w_i^{(\ell)} \cdot s$
        and~$|p|=i$ and~$|w_i^{(\ell)}|=\ell$.
        The walks $w'^{(\ell)}_i$, with $i\in[2k+2-\ell]$ and~$w''^{(\ell)}_j$, with $j\in[k+2]$ are defined similarly.
        It may be verified that~$\overrightarrow{O}_\ell$, given below, is $(2k+2)$-bounded.
        \begin{equation*}
            \overrightarrow{O}_\ell = \underbrace{w^{(\ell)}_{0}, w^{(\ell)}_{1}, \ldots, w^{(\ell)}_{2k+1-\ell}}_{\text{factors of~$w$}},~\underbrace{w'^{(\ell)}_{2k+1-\ell}, w'^{(\ell)}_{2k+1-\ell-1}, \ldots, w'^{(\ell)}_{0}}_{\text{factors of~$w'$}},~\underbrace{w''^{(\ell)}_{0}, w''^{(\ell)}_{1}, \ldots, w''^{(\ell)}_{k+1}}_{\text{factors of~$w''$}}
        \end{equation*}
        (In the case where~$\ell=2k+1$, we have~$\overrightarrow{O}_{2k+1}=w,w',w''^{(2k+1)}_{0}, w''^{(2k+1)}_{1}, \ldots, w''^{(2k+1)}_{k+1}$, which can also be verified to be $(2k+2)$-bounded.)
        
        If~$k$ is even, we write $O=(O_{\leq k},\overrightarrow{O}_{k+1},\overleftarrow{O}_{k+2},\ldots,\overrightarrow{O}_{2k+1}, O_{\geq 2k+2})$, where we choose the last walk of~$O_{\leq k}$ to be the prefix of length~$k$ of the first walk of $\overrightarrow{O}_{k+1}$ and the first walk of~$O_{\geq 2k+2})$ to be~$d^{2k+2}$.
        It follows that $G$ is $(2k+2)$-orderable by increasing length.
        If~$k$ is odd, the reasoning is similar with~$O=(O_{\leq k},\overleftarrow{O}_{k+1},\overrightarrow{O}_{k+2},\ldots,\overrightarrow{O}_{2k+1}, O_{\geq 2k+2})$ for a suitable choice of~$O_{\leq k}$.
    \end{proof}

    \begin{proof}[Proof of \cref{p:tight}\ref{p:tight:2p'-eve-incr-orderable}]
We consider the graph family shown in \cref{f:high-peripherality:eventual}
        and let~$H_k$ denote the graph of this family with (eventual) peripherality~$k$.
We first consider the case where~$k$ is an even number.
As shown in the figure,~$H_k$ has $2k+1$ vertices: $v_{-k},\ldots,v_0,\ldots,v_k$.
As we are discussing orderability, we will only consider the walks of length $\geq 2k$.
        Let us first consider the walks of length~$\ell\geq 2k$.
Note that there is exactly one walk in~$\walks[\ell]{H_k}$, denoted by~$w_{s,t}^{(\ell)}$, that starts at each~$v_s\in\{v_{-k},\ldots,v_0\}$ and ends at each~$v_t\in\{v_0,\ldots,v_k\}$.
Note also that, for every~$s\in\{-k,\ldots,0\}$ and~$t,t'\in\{0,\ldots,k\}$, we have ~$w_{s,t}^{(\ell)}\ppreleq{2i}w_{s,t'}^{(\ell)}$ where~$i=\max(t,t')$, hence also~$w_{s,t}^{(\ell)}\ppreleq{2k}w_{s,t'}^{(\ell)}$.
It follows that $\overrightarrow{O^{(\ell)}_s} \coloneq (w_{s,0},w_{s,1},\ldots, w_{s,k})$
        is a $2k$-bounded ordering of~$\{w_{s,t}\mid t\in\{0,\ldots,k\}\}$ for each $s$ in $\{-k,\ldots,0\}$. Denoting $\overleftarrow{O^{(\ell)}_s}$ the reverse ordering,
        let us consider the following ordering of~$\walks[\ell]{H_k}$.
        \begin{align}
            \overrightarrow{O^{(\ell)}} \coloneq \overrightarrow{O^{(\ell)}_{-k}},~\overleftarrow{O^{(\ell)}_{-k+1}},~\overrightarrow{O^{(\ell)}_{-k+2}},\overleftarrow{O^{(\ell)}_{-k+3}},~\ldots,~\overrightarrow{O^{(\ell)}_{0}} 
        \end{align}
        Note that~$\overrightarrow{O^{(\ell)}}$ always ends by $\overrightarrow{O^{(\ell)}_{0}}$ and never by~$\overleftarrow{O^{(\ell)}_{0}}$ since we assumed~$k$ to be an even number.
Since~$w_{s,k}^{(\ell)}\pprel{2k}w_{s+1,k}^{(\ell)}$ and~$w_{s,0}^{(\ell)}\pprel{2k}w_{s+1,0}^{(\ell)}$, for each~$s\in\{-k,\ldots,-1\}$, the ordering~$\overrightarrow{O^{(\ell)}}$ is~$2k$-bounded. 
The last walk of~$\overrightarrow{O^{(\ell)}}$ is~$w_{0,k}^{(\ell)}$ and the first walk of~$\overleftarrow{O^{(\ell+1)}}$ is~$w_{0,k}^{(\ell+1)}$, which satisfy~$w_{0,k}^{(\ell)}\pprel{1}w_{0,k}^{(\ell+1)}$.
Symmetrically, the last walk of~$\overleftarrow{O^{(\ell)}}$
        is~$w_{-k,0}^{(\ell)}$ 
        and the first walk of~$\overrightarrow{O^{(\ell+1)}}$ 
        is~$w_{-k,0}^{(\ell+1)}$,
        which also 
        satisfy~$w_{-k,0}^{(\ell)}\pprel{1}w_{-k,0}^{(\ell+1)}$.
It follows that~$O$, given below, is a $2k$-bounded ordering of~$\walks[\geq2k+1]{H_k}$.
        \begin{equation}
            O \coloneq \overrightarrow{O^{(2k+1)}},~\overleftarrow{O^{(2k+2)}},
            ~\ldots,~\overrightarrow{O^{(2\ell+1)}},~\overleftarrow{O^{(2\ell+2)}},~\ldots
        \end{equation}
        In other words, the graph~$H_k$ is eventually $2k$-orderable by increasing length when~$k$ is even.

        In the case where~$k$ is odd the ordering is defined in the same way.
        However, $\overrightarrow{O^{(\ell)}}$ ends with $\overleftarrow{O^{(\ell)}_{0}}$ hence with walk~$w_{0,0}^{(\ell)}$.
        Then $\overleftarrow{O^{(\ell+1)}}$ starts with~$w_{0,0}^{(\ell+1)}$, which is related by $\pprel{2k}$ (in fact even $\pprel{1}$) to~$w_{0,0}^{(\ell)}$,
        and ends with~$w_{-k,0}^{(\ell+1)}$.
    \end{proof}
\end{toappendix}

\section{Parameter Computation}
\label{s:param-comp}
\begin{toappendix}
    \label{a:param-comp}
\end{toappendix}

In this section, we show that inner-connectedness, peripherality, criticality, and their eventual versions can all be computed in polynomial time.
First, by computing the condensation in linear time (using, e.g., Tarjan's
algorithm~\cite{tarjan1972depth}), we obtain:

\begin{lemmarep}\label{p:type:linear}
    One may compute in linear time the kind of all objects in a graph.
\end{lemmarep}
\begin{proof}
    We first compute the strongly connected components of~$G$ and the condensation 
in linear time using Tarjan's algorithm~\cite{tarjan1972depth}.
    From this information, one may determine for every object whether it satisfies
    \cref{eq:kind:after-scc,eq:kind:before-scc}: for this we do one forward traversal of the graph starting
    from all nontrivial sccs, and one backward traversal of the graph
    starting from all nontrivial sccs. The information of which objects are
    reached by these traversals allows us to compute the kinds of all objects
    in linear time.
\end{proof}

We now turn to (eventual) peripherality, and show the next statement.

\begin{propositionrep}\label{p:peri-comp}
    Let~$G=(V,E)$ be a graph, and~$F$ the set of finitary vertices in~$G$.
    \begin{enumerate}[a.]
    \item \label{p:peri-comp:eve}
        The eventual peripherality of~$G$ can be computed in time~$O(|E|)$.
    \item \label{p:peri-comp:classic}
        The peripherality of~$G$ can be computed in time~$O(|E|\times (|F|+1))$.
    \end{enumerate}
\end{propositionrep}
\begin{proofsketch}
We first compute the kinds of objects thanks to \cref{p:type:linear}.
For every vertex~$v$, we define the following quantities:
\begin{itemize}
    \item $\ell(v)$, the maximal length of the $\L$-factor across walks
          of type~$\L^?\I^?$ ending at~$v$;
    \item $r(v)$, the maximal length of the $\R$-factor across walks 
          of type~$\I^?\R^?$ starting at~$v$;
\item  $M_k(v) \coloneq \max\big\{\,\ell(u)\;\big|\;
    u\in V \text{ s.t.\ there is a walk of type~\F{} of length~$k$ from~$u$ to~$v$}\,\big\}$, for every~$k$, $1\leq k \leq |F|+1$, fixing~$M_k(v) \coloneq-\infty$ if no such walk exists.
\end{itemize}
    Since the subgraph $(V,E')$ is acyclic, where~$E'=\{e\in E\mid e\text{ is not inner}\}$, these quantities are all well-defined. The first two can be computed in time~$O(|E|)$ by dynamic programming on the condensation.  The third one can be computed in~$O\big(|E|\times(|F|+1)\big)$ by increasing~$k$ again using dynamic programming.
    It may then be observed that the peripherality of $\L^?\I^?\R^?$-walks corresponds~to $\max_{v\in V} \min(\ell(v),r(v))$, for~$v\in V$,
    and that the  peripherality of 
$\L^?\F\R^?$-walks 
    corresponds to~$\max_{v\in V} \max_k(k+\min(M_k(v),r(v)))$.
\end{proofsketch}
\begin{proof}
We first compute the kind of every object in linear time 
using \cref{p:type:linear}.
Since every cycle entirely belongs to a nontrivial scc, the subgraph formed of the edges that are of kind~$\L$, $\R$ or $\F$ is acyclic.
For every vertex~$v$, we define two symmetric quantities:
\begin{itemize}
    \item $\ell(v)$, the maximal length of the $\L$-factor of a walk 
          of type~$\L^?\I^?$ ending at~$v$;
    \item $r(v)$, the maximal length of the $\R$-factor of a walk
      of type~$\I^?\R^?$ starting at~$v$.
\end{itemize}
The values~$\ell(v)$ are computed in linear time in two phases: dynamic programming over
the \L-subgraph first yields, for every vertex~$u$, the maximal length of a walk of
type~$\L^?$ ending at~$u$; these values are then propagated, taking maxima, along a
topological ordering of the condensation 
of the inner subgraph.
This indeed computes~$\ell$, since the walks of type~$\L^?\I^?$ ending at~$v$ are exactly
the concatenations of a walk of type~$\L^?$ ending at some vertex~$u$ with an inner walk
from~$u$ to~$v$, both possibly empty.
The values~$r(v)$ are computed symmetrically.

We will first show how to compute the peripherality of walks of type
$\L^?\I^?\R^?$ in \cref{c:peri-comp:lir}, and then treat the walks of type $\L^?\F\R^?$ in \cref{c:peri-comp:lfr}.
For walks of type $\L^?\I^?\R^?$ we distinguish two cases depending on whether we want to study peripherality or eventual peripherality:

\begin{claim}\label{c:peri-comp:lir}
The following equalities hold.
\begin{enumerate}[a.]
    \item \label{c:peri-comp:lir:full}
    $\max\big\{\peri(w)\;\big|\;w\text{ is of type $\L^?\I^?\R^?$}\big\}
    \;=\;    \max_{v\in V}\;\min\big(\ell(v),r(v)\big)$
    \item \label{c:peri-comp:lir:eve}
    $\max\big\{\peri(w)\;\big|\;w\text{ encounters a nontrivial scc}\big\}
    \;=\;
    \max_{v\in S}\;\min\big(\ell(v),r(v)\big)$, where $S$ denotes the set of the vertices that belong to a nontrivial scc.
\end{enumerate}
\end{claim}
 
\begin{claimproof}[Proof of \cref{c:peri-comp:lir}]
  We start with the ($\leq$) direction of~\eqref{c:peri-comp:lir:full}.  Let~$w=w_\ell\cdot w_i\cdot w_r$ be an $\L^?\I^?\R^?$-factorized walk.  
    We show that $\peri(w)$ is less than the right-hand-side of the claimed equality.
  If~$\peri(w)=0$, then the claim holds, hence we assume in the following that~$w_\ell$ and~$w_r$ are nonempty.
    In particular, there is a factorization~$w=w_1\cdot w_2$ such that~$w_1$ is of type~$\L\I^?$ 
    and~$w_2$ of type~$\I^?\R$. If we fix~$x=\tgt(w_1)$, we have~$|w_\ell|\leq\ell(x)$ and~$|w_r|\leq r(x)$, hence also~$\peri(w)\leq\min\big(\ell(x),r(x)\big)$.
    This immediately yields the ($\leq$) inequality of~\eqref{c:peri-comp:lir:full}.
    For the ($\leq$) inequality of~\eqref{c:peri-comp:lir:eve}, notice that any
    walk that encounters a nontrivial scc is of type~$\L^?\I^?\R^?$ hence the
    prior reasoning applies, picking the vertex~$x$ inside a nontrivial scc to be within the $\I^?$ factor or at the interface between the $\L^?$ and $\R^?$ factors, so achieving the maximum in the right-hand-side of equality~(\ref{c:peri-comp:lir:eve}) because $v$ can be chosen to be~$x$ which is indeed inside a nontrivial scc.

    \medskip
    
    We now show the ($\geq$) direction of \eqref{c:peri-comp:lir:full}. Let~$x$ be any vertex.
    Concatenating a walk that realizes~$\ell(x)$ and a walk that realizes~$r(x)$
    yields a walk~$w$ of type~$\L^?\I^?\R^?$ whose factorization has an $\L^?$-part of
    length~$\ell(x)$ and an $\R^?$-part of length~$r(x)$, hence of
    peripherality~$\min\big(\ell(x),r(x)\big)$.
    This yields the ($\geq$) inequality of~\eqref{c:peri-comp:lir:full}.
    Moreover, if~$x$ belongs to~$S$, then $w$~encounters a nontrivial scc, which yields the ($\geq$) inequality of~\eqref{c:peri-comp:lir:eve}.
\end{claimproof}

It remains to compute the maximal peripherality of $\L^?\F\R^?$-walks.
For every vertex~$v$ and integer~$k\geq1$, let
\begin{equation*}
    M_k(v) = \max\big\{\,\ell(u)\;\big|\; u \in V
    \text{s.t.\ there is a walk of type~\F{} of length~$k$ from~$u$ to~$v$}\,\big\},
\end{equation*}
with~$M_k(v) =-\infty$ if no such walk exists. Decomposing an \F-walk at its last
edge yields the recurrence
\begin{align*}
    M_1(v)&=\max\big\{\ell(\src(e))\;\big|\;e\text{ is an \F-edge},\ \tgt(e)=v\big\},\\
    M_{k+1}(v)&=\max\big\{M_k(\src(e))\;\big|\;e\text{ is an \F-edge},\ \tgt(e)=v\big\}.
\end{align*}
Since the \F-subgraph is acyclic, $M_k\equiv-\infty$ for all~$k>|F|+1$, so it suffices to consider values $1 \leq k \leq |F|+1$ and all values are
computed in time~$O(|E|\times(|F|+1))$.

\begin{claim}\label{c:peri-comp:lfr}$\displaystyle\max\big\{\peri(w)\,\big|\,w\text{ is of type $\L^?\F\R^?$}\big\}
    =
    \max_{\strut v\in V,1 \leq k\leq |F|+1}\big\{\,k+\min(M_k(v),r(v))\big\}$
\end{claim}
\begin{claimproof}[Proof of \cref{c:peri-comp:lfr}]
($\leq$) Let~$w=w_\ell\cdot w_f\cdot w_r$ be an $\L^?\F\R^?$-factorized
walk, and let~$k=|w_f|$, $u=\src(w_f)$ and~$v=\tgt(w_f)$.
Let us show that $\peri(w)$ is less than the right-hand-side.
Then~$|w_\ell|\leq\ell(u)\leq M_k(v)$ and~$|w_r|\leq r(v)$,
and the following equation concludes the proof of ($\leq$).
\begin{equation*}
    \peri(w)=k+\min(|w_\ell|,|w_r|)\leq k+\min\big(M_k(v),r(v)\big)
\end{equation*}

\medskip

($\geq$) Let~$v$ and~$k\geq 1$ be such that~$M_k(v)\neq-\infty$, realized by a vertex~$u$ and
an \F-walk~$w_f$ from~$u$ to~$v$. Moreover, we let~$w_\ell$ and~$w_r$ be some walks that realize~$\ell(u)$ and~$r(v)$.
The walk~$w'=w_\ell \cdot w_f \cdot w_r$ cannot contain both an
\I-edge and an \F-edge from \cref{pp:walk-type}\ref{pp:walk-type:no-i&f}.
Since~$w'$ contains an \F-edge, it is necessarily of type $\L^?\F\R^?$ and~$w_\ell \cdot w_f \cdot w_r$ is its $\L^?\F\R^?$-factorization, hence: 
\begin{equation*}
    \peri(w') = k+\min\big(\ell(u),r(v)\big) = k+\min\big(M_k(v),r(v)\big)
\end{equation*}
which concludes the proof of $(\geq)$ hence of \cref{c:peri-comp:lfr}.
\end{claimproof}

The eventual peripherality of~$G$ is equal to the quantity in Item~\ref{c:peri-comp:lir}\ref{c:peri-comp:lir:eve} from \cref{p:peri-direct}\ref{p:peri-direct:eve}.
Since~$\ell(v)$ and $r(v)$ for all~$v \in V$ are computed in $O(|E|)$ total time, \cref{p:peri-comp:eve} of \cref{p:peri-comp} follows.
The peripherality of~$G$ is the maximum of the quantities in Items~\ref{c:peri-comp:lir}\ref{c:peri-comp:lir:full} and~\ref{c:peri-comp:lfr}.
Since all~$M_k(v)$ are computed for all~$v\in V$ and~$k\in\N$ in $O(|E|\cdot(|F|+1))$ total time, \cref{p:peri-comp:classic} of \cref{p:peri-comp} follows.
\end{proof}

We now show that the criticality and eventual criticality can be computed in linear time,
similarly to the proof of~\cref{p:peri-comp}\ref{p:peri-comp:eve}. Indeed, 
we simply compute the maximal length of a walk of type~$\L^?\F^?$ (resp.\ of type~$\F^?\R^?$) ending in
(resp.~starting in) each vertex~$v$.

\begin{propositionrep}\label{p:crit-comp}
    Given a graph~$G=(V,E)$, the criticality and eventual criticality of~$G$ can
    be computed in time~$O(|E|)$.
\end{propositionrep}

\begin{proof}
We first discuss criticality.  From the definition (\cref{d:criticality}), the criticality is the maximum between the criticality of walks of type $\L^?\I^?\R^?$ and those of type $\L^?\F\R^?$.
The maximal criticality of walks of type $\L^?\I^?\R^?$ can be computed in linear time as in the proof of
\cref{p:peri-comp} using the result of
    \cref{c:peri-comp:lir}\ref{c:peri-comp:lir:full}, because the maximal criticality over such walks is obtained as the maximal peripherality minus one. 
    Now, the maximal criticality of walks of type $\L^?\F\R^?$ can be computed as follows.
    Like in the beginning of the proof of \cref{p:peri-comp}, we first compute the kinds of all objects using \cref{p:type:linear}, and we use the fact that the subgraph of edges of kind $\L$, $\R$ and $\F$ is acyclic. This allows us to compute in linear time over this graph by dynamic programming the following quantities for every vertex $v$: the maximal length of a walk of type $\L^?\F^?$ that ends at~$v$, and the maximal length of a walk of type $\F^?\R^?$ that starts at~$v$. Once this is done, the maximal criticality of walks of type $\L^?\F\R^?$ is obtained following \cref{d:criticality} by taking the maximum, over all edges $(u,v)$ of kind $\F$, of the minimum of the first quantity for~$u$ and of the second quantity for~$v$.

   For eventual criticality, we use again that walks that encounter a nontrivial scc must be of type $\L^?\I^?\R^?$, so that we compute it from peripherality again simply using  \cref{c:peri-comp:lir}\ref{c:peri-comp:lir:eve}. This concludes the proof.
\end{proof}

Recall that the peripherality of graphs never exceeds twice the criticality plus
two (\cref{pp:peripherality-vs-criticality}), so
peripherality can be 2-approximated by criticality plus one, in linear time.

We now show that, by contrast, it is NP-hard to determine whether the pp-increasing-rank of a given graph~$G$ is no greater than~$4$:

\begin{propositionrep}
  \label{p:4ord-nphard}
    Deciding whether a graph is 4-orderable by increasing length is NP-hard, even for acyclic graphs.
\end{propositionrep}

\begin{proofsketch}
We reduce from the Edge Hamiltonian Path Problem on
  bipartite graphs~\cite{LaiWei1993}. In this problem, we are given as input
  an undirected bipartite graph~$G = (U \uplus V, E)$, and we ask whether the 
  \emph{line graph} $\mathrm{L}(G)$ of~$G$
  has a Hamiltonian path. Recall that $\mathrm{L}(G)$ is the graph on~$E$ where
  we connect any two edges that share an endpoint.
We then build an acyclic graph $G'$ from~$G$ by orienting its edges from~$U$ to~$V$ and adding a source $s$ connected to every vertex of~$U$ and a target~$t$ connected to every vertex of~$V$. This ensures that walks of length~$3$ in~$G'$ correspond to
  edges of~$G$, and that two such walks are at push-pop distance of 4 precisely
  when the corresponding edges of~$G$ share an endpoint. Thus, a 4-ordering of~$G'$ exists precisely when
  $\mathrm{L}(G)$  admits a Hamiltonian path.
\end{proofsketch}

\begin{proof}
  Recall that the \emph{line graph} of an undirected graph $G$ is the
  graph $G'$ whose vertex set consists of the edges of~$G$ where we connect two
  edges of~$G$ in~$G'$ whenever they
  share a vertex in~$G$.
  We reduce from the Edge Hamiltonian Path Problem~\cite{LaiWei1993} where we
  ask, given an undirected graph $G$ as input, whether it admits an
  \emph{edge Hamiltonian path}, which is simply a Hamiltonian path of the line
  graph of~$G$. The Edge Hamiltonian Path Problem is NP-hard even when $G$ is required
  to be a bipartite graph~\cite{LaiWei1993}.

  To describe the reduction, let $G = (A \uplus B, E)$ be the
  bipartite graph that we are given as input. We construct the directed graph $G' = (\{s,t\} \uplus A \uplus B, E')$ as follows:
    \begin{itemize}
        \item There is an edge from $s$ to each vertex in $A$;
        \item There is an edge from each vertex in $B$ to $t$;
        \item For each $e\in E$, there is an edge $e'\in E'$ such that $\{\src(e), \tgt(e)\} = \{\src(e'), \tgt(e')\}$ and $\src(e') \in A$, $\tgt(e')\in B$ (that is, we direct every edge of $G$ from $A$ to $B$).
    \end{itemize}
    It is immediate that the construction of~$G'$ is in polynomial time in the
    input~$G$.

    Remark that every walk of length $3$ in $G'$ is of the form $s \rightarrow a \rightarrow b \rightarrow t$, with $a\in A$ and $b\in B$. Thus, there is a bijection between the walks of length $3$ of $G'$ and the edges of $G$.

    Finally, remark that two walks of length $3$ are related by $\ppreleq{4}$ if and only if they share their vertex in $A$ or their vertex in $B$. Formally, let $w = s\xrightarrow{e_1}a\xrightarrow{e_2}b\xrightarrow{e_3}t$ and $w' = s\xrightarrow{e_1'}a'\xrightarrow{e_2'}b'\xrightarrow{e_3'}t$. 
    Then we have $w \ppreleq{4} w'$ if and only if $e_1 = e_1'$ or $e_3 = e_3'$, which yields $a = a'$ or $b = b'$.
    This is precisely the condition under which two edges of~$G$ are
    adjacent in the line graph of~$G$, i.e., 
    the condition under which
two edges of $G$ can occur
    consecutively in an edge Hamiltonian path. Thus there exists a
    $4$-ordering of the walks of length $3$ of $G'$ if and only if $G$
    admits
an edge Hamiltonian path.

    This concludes the proof. Indeed, walks of length $1$ or $2$ of $G'$ are
    always $4$-orderable. Thus $G'$ is $4$-orderable by increasing length if and
    only if $G$ has an edge Hamiltonian path.
\end{proof}

We note that the proof focuses on walks of small length, so it does not generalize to eventual pp-increasing-rank. We do not know whether hardness generalizes to this parameter, and we also do not know complexity upper bounds: we will discuss this in the conclusion.

\section{Conclusion and Future Directions}
\label{s:conclusion}

Our work introduces the question of finding walk orderings of a graph that bound the push-pop distance between consecutive walks.
Our main contribution is threefold. 
First, we proved that the pp-rank of inner-connected graphs is bounded by 3, improving upon the previous bound of $48(|V|+1)^2$.
Second, we introduced two graph measures, peripherality and criticality, that yield almost tight lower bounds (sometimes up to an additive constant of 1) and approximate upper bounds (up to a constant factor) to determine the pp-increasing-rank and eventual pp-increasing-ranks.
Third, we showed that peripherality and criticality can be computed in polynomial time, while computing the exact pp-increasing-rank is intractable.

Our third contribution leaves open several complexity questions.
Indeed, we have no lower bound on the complexity of computing the pp-rank or the eventual pp-increasing-rank.
More surprisingly, we have no complexity upper bounds on these tasks at all, even to show that they are computable.
We give in Appendix~\ref{a:open-problems} a more complete list of open questions.

Further, 
our initial motivation from the introduction is to 
efficiently enumerate walks in graph databases.  
Inner-connectedness 
determines the \emph{existence} of 
a bounded walk enumeration, 
but 
one would like to design \emph{effective} algorithms with good bounds on the \emph{delay} between two consecutive walks and on the \emph{memory usage}.
We have hope that 
such algorithms can be designed in the
model of deque machines from~\cite{CDCASGrayCode-icalp}.
This would be a substantial improvement over~\cite{amarilli2023enumerating},
which needs unbounded auxiliary memory.
Another 
question is whether walk enumeration on non-inner-connected graphs
can be done with a limited number of \emph{threads} like in~\cite{amarilli2023enumerating}; we expect that our results would generalize but leave this open.

In a different direction, a natural extension of our work would be to allow additional elementary operations beyond push and pop. 
For instance, we could allow the insertion and deletion of edges at arbitrary positions of the walk, leading to a distance on walks amounting to the Levenshtein distance on words. 
This raises the question of determining which graphs are Levenshtein-orderable and what is their Levenshtein-rank.
One can see by inspecting the proof of \cref{p:inner-disconnected=>nonorderable} (Appendix~\ref{a:typology}) or of \cite[Theorem 3.9]{amarilli2023enumerating} that inner-connectedness remains a necessary condition for Levenshtein-orderability, and of course Levenshtein-rank is upper bounded by pp-rank so it is always at most~$3$ for inner-connected graphs. 
However, when ordering walks in increasing order of length, the situation is different, because the Levenshtein-increasing-rank (defined in the expected way) no longer admits a lower bound from peripherality and criticality. 
Indeed, \cref{f:levenstein} shows a family of graphs with unbounded peripherality and bounded Levenshtein-increasing-rank.
Indeed, notice that each walk $e_1\,e_2\cdots e_i \,f_{i+1}\,g^j\,w $ is at Levenshtein distance $4$ from~$e_1\,e_2\,\cdots e_{i-1}\,f_{i}\,g^{j+1}\,w$, so we could use \cref{p:layer:even-odd} to upper bound the eventual Levenshtein-increasing-rank by~$12$.
Note that if we replace the loop labeled by $g$ with a simple cycle of length~$k$, then the Levenshtein increasing-rank is no longer bounded: this suggests that entirely new tools would be necessary to extend our results to other distances such as the Levenshtein distance.

\begin{figure}[ht]
    \centering
    \begin{tikzpicture}
    \newcommand{\dist}{14mm}\input{figure/graphstyle}

    \path node[vertex,outer] (v1) {} 
            ++(0:\dist) node [vertex, outer] (v2){}
            ++(0:\dist) node [vertex, outer,opacity=0] (v3){}
            ++(0:1.25*\dist) node [vertex, outer,opacity=0] (vkmo){}
            ++(0:\dist) node [vertex, outer] (vk) {}
            ++(0:\dist) node [vertex, inner] (r) {}
            ++(0:\dist) node [vertex, outer] (v1p) {}
            ++(0:\dist) node [vertex, outer,opacity=0] (v2p) {}
++(0:1.25*\dist) node [vertex, outer,opacity=0] (vkmop) {}
            ++(0:\dist) node [vertex, outer] (vkp) {}
    ;

    \path[edge,outer] (v1) to node[label, near start,above]{$e_{1}$} (v2);
    \path[edge,outer] (v2) to node[label, near start,above]{$e_{2}$} (v3);
    \path (v3) to node[midway,outer] {$\cdots$} (vkmo);
    \path[edge,outer] (vkmo) to node[label, near start,above]{$e_{k}$} (vk);
    \path[edge,outer] (vk) to node[label, pos=.4,above]{${f_{k+1}}$} (r);
    \path[edge,outer] (r) to (v1p);
    \path[edge,outer] (v1p) to (v2p);
    \path (v2p) to node[midway,outer] {$\cdots$} (vkmop);
    \path[edge,outer] (vkmop) to (vkp);
    
    \path[edge,inner,north loop,overlay] (r) to node[label, above](rlabel)
      {$g$} (r);

    \begin{pgfinterruptboundingbox}
        \path[edge,outer,overlay] (v1) to[out=-20,in=-90,looseness=.6] node[label, near start, below]{${f_{1}}$} (r);
        \path[edge,outer] (v2) to[out=-20,in=-120,looseness=.7] node[label, near start, above right,yshift=-1mm]{${f_{2}}$} (r);
        \path[edge,outer] (vkmo) to[out=-20,in=-150] node[label, pos=.1, below left, xshift=3mm]{${f_{k}}$}(r);
    \end{pgfinterruptboundingbox}
    
    \draw [decorate,decoration={brace,amplitude=1em,mirror,raise=2em}]
  (r) -- (vkp) node[midway,yshift=-4em]{Some walk~$w$ s.t.~$|w|=k+1$};
    
\end{tikzpicture}
    \caption{Graphs of unbounded eventual peripherality and bounded Levenshtein increasing rank}
    \label{f:levenstein}
\end{figure}

\section{Bibliography}
\label{sec:bibliography}
\printbibliography[heading=none]

@inproceedings{MartensEtAl24,
  author       = {Wim Martens and
                  Matthias Niewerth and
                  Tina Popp and
                  Carlos Rojas and
                  Stijn Vansummeren and
                  Domagoj Vrgoc},
  title        = {Compact path representations for graph database pattern matching},
  booktitle    = {ICDE'24},
  -pages        = {379--380},
  -publisher    = {{IEEE}},
  year         = {2024},
  doi          = {10.1109/ICDEW61823.2024.00059},
}

@article{MartensEtAl23,
  author       = {Wim Martens and
                  Matthias Niewerth and
                  Tina Popp and
                  Carlos Rojas and
                  Stijn Vansummeren and
                  Domagoj Vrgoc},
  title        = {Representing paths in graph database pattern matching},
  journal      = {Proc. {VLDB} Endow.},
  volume       = {16},
  number       = {7},
  -pages        = {1790--1803},
  year         = {2023},
  doi          = {10.14778/3587136.3587151},
}

@inproceedings{LibkinEtAl25-gems,
  author       = {Leonid Libkin and
                  Wim Martens and
                  Filip Murlak and
                  Liat Peterfreund and
                  Domagoj Vrgoc},
  title        = {Querying graph data: Where we are and where to go},
  booktitle    = {Companion of {PODS}'25},
  -pages        = {9--26},
  publisher    = {{ACM}},
  year         = {2025},
  doi          = {10.1145/3722234.3725822},
}

@article{LiangEtAl2024,
author = {Liang, Qi and Ouyang, Dian and Zhang, Fan and Yang, Jianye and Lin, Xuemin and Tian, Zhihong},
title = {Efficient regular simple path queries under transitive restricted expressions},
year = {2024},
issue_date = {March 2024},
journal = {VLDB Endowment},
volume = {17},
number = {7},
issn = {2150-8097},
doi = {10.14778/3654621.3654636},
-pages = {1710–1722},
numpages = {13},
keywords = {enumeration, reachability, regular expression, simple path}
}

@inproceedings{amarilli2023enumerating,
  author = {Antoine Amarilli and Mikaël Monet},
  title = {\href{https://arxiv.org/abs/2209.14878}
    {Enumerating regular languages with bounded delay}},
  year = {2023},
  booktitle = {\mbox{\href{https://www.univ-orleans.fr/lifo/stacs/}{STACS}}},
  doi = {10.4230/LIPIcs.STACS.2023.8},
}

@inproceedings{adamson2024enumerating,
  title={Enumerating $m$-length walks in directed graphs with constant delay},
  author={Adamson, Duncan and Gawrychowski, Pawe{\l} and Manea, Florin},
  booktitle={LATIN},
  doi          = {10.1007/978-3-031-55598-5_3},
  year={2024},
  -publisher={Springer}
}

@article{Eppstein98-compact-shortest,
  author       = {David Eppstein},
  title        = {Finding the $k$ shortest paths},
  journal      = {{SIAM} J. Comput.},
  volume       = {28},
  number       = {2},
  -pages        = {652--673},
  year         = {1998},
  doi          = {10.1137/S0097539795290477},
}

@inproceedings{CDCASGrayCode-icalp,
  author       = {Antoine Amarilli and
                  Claire David and
                  Nadime Francis and
                  Victor Marsault and
                  Mika{\"{e}}l Monet and
                  Yann Strozecki},
  title        = {Gray codes with constant delay and constant auxiliary space},
    booktitle = {ICALP'26},
    -series = {LIPIcs},
    ISBN = {978-3-95977-428-4},
    ISSN = {1868-8969},
    volume =  {374}, 
    -publisher = {Schloss Dagstuhl -- Leibniz-Zentrum f{\"u}r Informatik},
    year    = {2026},
    doi     = {10.4230/LIPIcs.ICALP.2026.160}
}

@inproceedings{FrancisEtAl2018,
  author       = {Nadime Francis and
                  Alastair Green and
                  Paolo Guagliardo and
                  Leonid Libkin and
                  Tobias Lindaaker and
                  Victor Marsault and
                  Stefan Plantikow and
                  Mats Rydberg and
                  Petra Selmer and
                  Andr{\'{e}}s Taylor},
  -editor       = {Gautam Das and
                  Christopher M. Jermaine and
                  Philip A. Bernstein},
  title        = {Cypher: An evolving query language for property graphs},
  booktitle    = {SIGMOD'18},
  -pages        = {1433--1445},
  publisher    = {{ACM}},
  year         = {2018},
  url          = {https://doi.org/10.1145/3183713.3190657},
  doi          = {10.1145/3183713.3190657},
}

@article{tarjan1972depth,
  title={Depth-first search and linear graph algorithms},
  author={Tarjan, Robert},
  journal={SIAM journal on computing},
  volume={1},
  number={2},
  -pages={146--160},
  year={1972},
  publisher={SIAM}
}

@article{sekanina1960ordering,
  title={On an ordering of the set of vertices of a connected graph},
  author={Sekanina, Milan},
  journal={Publ. Fac. Sci. Univ. Brno},
  volume={412},
  -pages={137--142},
  year={1960}
}

@article{karaganis1968cube,
  title={On the cube of a graph},
  author={Karaganis, Jerome J.},
  journal={Canadian Mathematical Bulletin},
  volume={11},
  number={2},
  -pages={295--296},
  year={1968},
  -publisher={Cambridge University Press},
  doi={10.4153/CMB-1968-037-0}
}

@article{feder1994network,
  title={Network flow and 2-satisfiability},
  author={Feder, Tom{\'a}s},
  journal={Algorithmica},
  volume={11},
  number={3},
  -pages={291--319},
  year={1994},
  -publisher={Springer},
  doi={10.1007/BF01240738}
}

@article{Mutze23-gray, 
    title   = {Combinatorial Gray codes — An updated survey}, 
    volume  = {Dynamic Surveys}, 
    doi     = {10.37236/11023}, 
    journal = {The Electronic Journal of Combinatorics}, 
    author  = {Mütze, Torsten}, 
    year    = {2023}, 
}

@article{chung1992universal,
  title={Universal cycles for combinatorial structures},
  author={Chung, Fan and Diaconis, Persi and Graham, Ron},
  journal={Discrete Mathematics},
  volume={110},
  number={1-3},
  -pages={43--59},
  year={1992},
  -publisher={North-Holland},
  doi={10.1016/0012-365X(92)90699-G},
}

@article{LaiWei1993,
    title = {The edge Hamiltonian path problem is NP-complete for bipartite graphs},
    journal = {Information Processing Letters},
    volume = {46},
    number = {1},
    -pages = {21-26},
    year = {1993},
    issn = {0020-0190},
    doi = {https://doi.org/10.1016/0020-0190(93)90191-B},
    author = {Ten-Hwang Lai and Shu-Shang Wei},
}

@article{ng2017state,
  title={State complexity of prefix distance},
  author={Ng, Timothy and Rappaport, David and Salomaa, Kai},
  journal={Theoretical Computer Science},
  volume={679},
  -pages={107--117},
  year={2017},
  publisher={Elsevier}
}

@inproceedings{segoufin2014glimpse,
  title={A glimpse on constant delay enumeration (invited talk)},
  author={Segoufin, Luc},
  booktitle={STACS},
  year={2014},
  -organization={Schloss Dagstuhl-Leibniz-Zentrum f{\"u}r Informatik}
}

@unpublished{ruskey2003combinatorial,
  title={Combinatorial generation},
  author={Ruskey, Frank},
  note={Working Version (1j-CSC 425/520)},
  -volume={11},
  -pages={20},
  year={2003},
  url = {https://page.math.tu-berlin.de/~felsner/SemWS17-18/Ruskey-Comb-Gen.pdf}
}

@article{david2024distinct,
  title={Distinct shortest walk enumeration for RPQs},
  author={David, Claire and Francis, Nadime and Marsault, Victor},
  journal={Proceedings of the ACM on Management of Data},
  volume={2},
  number={2},
  -pages={1--22},
  year={2024},
  publisher={ACM New York, NY, USA}
}

@book{harary1969,
  author    = {Harary, Frank},
  title     = {Graph theory},
  publisher = {Addison-Wesley},
  -address   = {Reading, MA},
  year      = {1969}
}

@inproceedings{martens2018evaluation,
  title={Evaluation and enumeration problems for regular path queries},
  author={Martens, Wim and Trautner, Tina},
  booktitle={ICDT},
  year={2018},
  -organization={Schloss Dagstuhl-Leibniz-Zentrum f{\"u}r Informatik}
}

@misc{GQL-ISO,
    shorthand = {ISO24},
    author = {{International Organization for Standardization}},
    title = {{GQL}},
    howpublished = {\url{https://www.iso.org/standard/76120.html}},
    note = {Standard ISO/IEC CD 39075}, 
    year= {2024}
}

@misc{SQLPGQ-ISO,
    shorthand = {ISO23},
    author = {{ISO/IEC JTC 1/SC 32}},
    title = {{Part 16: Property Graph Queries (SQL/PGQ)}},
    howpublished = {\url{https://www.iso.org/standard/79473.html}},
    note = {ISO/IEC 9075-16:2023. Information technology – Database
languages SQL}, 
    year= {2023}
}

\end{document}